\documentclass[a4paper]{cas-dc}
\usepackage[numbers]{natbib}

\usepackage{preamble}
\usepackage{macros}

\begin{document}
\let\WriteBookmarks\relax
\def\floatpagepagefraction{1}
\def\textpagefraction{.001}
\shorttitle{Explicit Second-Order Bounds on the Domain of Validity for Lyapunov--Schmidt Reduction}
\shortauthors{P. Gupta et~al.}

\title[mode=title]{Explicit Second-Order Bounds on the Domain of Validity for Lyapunov--Schmidt Reduction}

\author[1,2]{Pranav Gupta}[orcid=0009-0004-2965-8196]
\ead{guptapranav@iitb.ac.in}
\cormark[1]
\cortext[1]{Corresponding Author}

\author[1]{Ravi Banavar}[orcid=0000-0002-5746-7096]
\ead{banavar@iitb.ac.in}

\author[2]{Anastasia Bizyaeva}[orcid=0000-0003-0728-9655]
\ead{anastasiab@cornell.edu}

\affiliation[1]{
    organization={Centre for Systems and Control, Indian Institute of Technology Bombay},
    city={Mumbai},
    citysep={},
    postcode={400076},
    state={Maharashtra},
    country={India}
}

\affiliation[2]{
    organization={Sibley School of Mechanical and Aerospace Engineering, Cornell University},
    city={Ithaca},
    citysep={},
    postcode={14853},
    state={New York},
    country={United States of America}
}


\begin{abstract}
    Lyapunov--Schmidt reduction is a widely used dimensionality reduction technique for bifurcation analysis in high-dimensional systems. While classical formulations guarantee the local existence of reduced-order equations, they typically lack explicit quantitative estimates on the size of the neighbourhoods in which these representations faithfully capture the full bifurcation structure. In recent work, we have addressed this limitation by deriving bounds on the domain of validity of this reduction using first-order conditions on the vector field together with a quantitative result for the implicit function theorem. In this article we explore beyond, and adopt second-order conditions that incorporate the Hessians of the vector field to develop these bounds, and obtain a new set of results. We then investigate the applicability of these newly derived bounds to Hopfield-like networked dynamical systems, evaluate these bounds for pitchfork bifurcations on connected regular graphs, and finally examine the relationship between the two certified domains for this class of systems.
\end{abstract}

\begin{highlights}
    \item Derive explicit second-order validity bounds for Lyapunov--Schmidt reduction.
    \item Admissible radii admit an explicit closed-form geometric characterisation.
    \item Bounds evaluated in closed form for Hopfield networks on connected regular graphs.
    \item Second-order bounds proved contained in published first-order bounds for regular graphs.
    \item Simple counterexample shows that neither bound family dominates the other in general.
\end{highlights}

\maketitle

\section{Introduction}
\label{sec:introduction}

Local bifurcations are points of transition in the qualitative behaviour of nonlinear dynamical systems, encompassing fundamental changes in the nature of equilibria - stability and count, emergence of limit cycles, and multistability - arising as system parameters cross critical thresholds~\cite{guckenheimer2013nonlinear}. Bifurcation analysis plays a central role in the design and regulation of such behaviours across a wide range of applied systems, including ecology \cite{scheffer2009early,dai2012generic}, climate dynamics \cite{ashwin2012tipping,lenton2008tipping}, cell biology \cite{tyson2003sniffers,novak2008design}, computational neuroscience \cite{izhikevich2000neural}, cardiac dynamics \cite{qu2014nonlinear}, and transportation \cite{orosz2010traffic}.

As the dimensionality of a dynamical system grows, analytical treatment of the vector fields involved grows increasingly difficult. A key step in rendering bifurcation analysis of high-dimensional systems tractable is dimensionality reduction, since bifurcations generically occur on low-dimensional manifolds \cite{pliss1964reduction,kelley1966stable}. The \LSR{} is a widely used tool for this purpose \cite{golubitsky2023dynamics,sidorov2013lyapunov}, reducing the problem to a lower-dimensional algebraic system whose zeros correspond one-to-one with equilibria of the full system near the bifurcation point. This construction is, however, intrinsically local: the size of the neighbourhood over which the reduced equations faithfully represent the full bifurcation diagram is typically unknown. Explicit bounds addressing this were first presented in \cite{gupta2024estimates}, building on \ImFT{} domain estimates from \cite{jindal2024estimates}, and specialised in \cite{gupta2025bounds} to Hopfield-type networked dynamical systems, evaluated explicitly for connected regular graphs.

The previously derived first-order bounds use one of two formulations of the \ImFT{} validity bounds in \cite{jindal2024estimates}; the alternative instead incorporates Hessians alongside Jacobians, giving a validity condition with a different computational structure. Because the two formulations follow distinct proof strategies, neither is guaranteed a priori to certify the larger neighbourhood, and we regard the resulting bounds as complementary rather than successive refinements of one another. This second-order structure further offers gains the first-order bounds lack: an admissible neighbourhood determined by explicit algebraic conditions that are numerically simpler and geometrically characterisable, and which may themselves be relaxed to readily computable conservative estimates without forfeiting the validity guarantee. This article develops these \emph{second-order} bounds for the general \LSR{} setting, specialises them to the same networked systems considered previously, and shows that, for this class, the resulting domain is contained within the first-order bound. We further show that this inclusion need not hold in general, that the regular-graph result is a special case of a broader degeneracy condition, and finally contrast the computational and geometric structure of the two bounds families.

The remainder of this paper is organised as follows. \Cref{sec:background} establishes the necessary notation and recalls the foundational prior results on explicit second-order domain estimates for the \ImFT{} and the formulation of the \LSR{}. \Cref{sec:lsr2} presents the main contribution: explicit second-order bounds on the validity domain of the \LSR{} and characterises the resulting admissible region. \Cref{sec:NLD} applies these bounds to Hopfield-like networked systems, and evaluates them analytically at pitchfork bifurcations on connected regular graphs. \Cref{sec:1v2} provides a quantitative and qualitative comparison of the second-order bounds with the first-order bounds established in prior work. \Cref{sec:numerics} presents numerical studies that illustrate the theoretical results, and supporting derivations are collected in \Cref{app:dv,app:Q}.

\section{Background}
\label{sec:background}

\subsection{Notation and Definitions}
\label{ssec:notations}


For \(z \in \R[n]\), we use \(\mathcal{N}(z) \subset \R[n]\) to denote an open neighbourhood of \(z\). For \(z_0 \in \R[n]\) and \(r > 0\), the open ball of radius \(r\) centred at \(z_0\) is given by \(\mathcal{B}(z_0,r) := \{z \in \R[n] : \norm{z - z_0}_2 < r\}\), where \(\norm{\cdot}_2\) is the Euclidean \(2-\)norm on \(\R[n]\). We denote the spectral and Frobenius norms of an arbitrary tensor \(A\) by \(\norm{A}_2\) and \(\norm{A}_F\) respectively.
For a given rank\(-3\) tensor \(T \in \R[n_1 \times n_2 \times n_3]\) and rank\(-2\) tensors \(U \in \R[n_2 \times m_2]\) and \(V \in \R[n_3 \times m_3]\), let \(T[U,V]\) denote the bilinear contraction of \(T\) with \(U\) and \(V\), defined as \(\bigl[T[U,V]\bigr]_{ijk}:=\sum_{\mu,\nu}T_{i\mu\nu}U_{\mu j}V_{\nu k}\).
We write \(\0[n], \1[n] \in \R[n]\) for the all-zero and all-one column vectors respectively, and \(\0[p \times q], \1[p \times q] \in \R[p \times q]\) for their matrix versions. The order \(n\) identity matrix is denoted by \(I_n \in \R[n \times n]\). For a rank\(-3\) tensor \(T\in\R[l \times (m_1+m_2) \times (n_1+n_2)]\), the notation
\begin{equation*}
    T = \begin{bNiceMatrix}[margin=0.6ex,hvlines]A & B \\ C & D\end{bNiceMatrix}_{i}
\end{equation*}
denotes a block representation of \(T\), where the \(j\) and \(k\) index ranges are partitioned into row and column blocks \(A\in\R[l \times m_1 \times n_1], B\in\R[l \times m_1 \times n_2], C\in\R[l \times m_2 \times n_1], D\in\R[l \times m_2 \times n_2]\), while the \(i\)-index is shared by all blocks.

Let \(\mathcal{G} = (\mathcal{V}, \mathcal{E})\) be a graph with vertex set \(\mathcal{V}\) of \(n\) nodes, and edge set \(\mathcal{E} \subseteq \mathcal{V} \times \mathcal{V}\). We associate with \(\mathcal{G}\) the signed adjacency matrix \(A \in \R[n \times n]\), where \(a_{ij}\) is the signed weight of the edge between nodes \(i\) and \(j\), and \(a_{ij} = 0\) if no such edge exists. A graph is \emph{\(k\)-regular} if every vertex has degree \(k\). Equivalently, the binary adjacency matrix \(A\) (unit weight on existing edges) of a \(k\)-regular graph satisfies \(A \mathbf{1}_{n} = k \mathbf{1}_{n}\).

\subsection{\ImFT{}}
\label{ssec:imft}

\begin{theorem}
    \label{thm:imft}
    \cite[Theorem~4.B]{zeidler2013vol1}
    Let \(f : U \times V \to \R[m]\) be a smooth map, where \(U \subset \R[n]\) and \(V \subset \R[m]\) are open sets. Consider the system of equations \(f(u,v) = f(u_0,v_0)\) for some \((u_0,v_0) \in U \times V\) for which \(\D[v]f(u_0,v_0)\) is invertible. Then, there exist open sets \(\mathcal{N}(u_0) \subset U\) and \(\mathcal{N}(v_0) \subset V\) equipped with a unique mapping \(g : \mathcal{N}(u_0) \to \mathcal{N}(v_0)\) such that \(g(u_0)=v_0\) and \(f\bigl(u,g(u)\bigr)=f(u_0,v_0), \quad \forall \; u\in\mathcal{N}(u_0)\).
\end{theorem}

\begin{theorem}
    \label{thm:imft2}
    Restating \cite[Theorem 3.3]{jindal2024estimates}, we define the following quantities:
    \begin{equation}
        M_{u} := \norm{\D[u]f(u_0,v_0)}, \; M_{v} := \norm{\bigl[\D[v]f(u_0,v_0)\bigr]^{-1}}.
        \label{def:imft2-Ms}
    \end{equation}
    For given \(R_u,R_v>0\) such that \(\mathcal{B}(u_0,R_u) \subset U\) and \(\mathcal{B}(v_0,R_v) \subset V\), let \(\Omega := \mathcal{B}(u_0,R_u)\times\mathcal{B}(v_0,R_v)\). Define%
    \begin{subequations}
        \begin{align}
            \Kuu & :=\sup\left\{\norm{\D[u]^2f(u,v)}    : (u,v) \in \Omega\right\} \label{def:K_uu}, \\
            \Kuv & :=\sup\left\{\norm{\D[u]\D[v]f(u,v)} : (u,v) \in \Omega\right\} \label{def:K_uv}, \\
            \Kvv & :=\sup\left\{\norm{\D[v]^2f(u,v)}    : (u,v) \in \Omega\right\} \label{def:K_vv}.
        \end{align} \label{def:Ks}
    \end{subequations}
    Then for all \(0<r_u<R_u\) and \(0<r_v<R_v\) satisfying
    \begin{subequations}
        \begin{align}
            \frac12 \Kuu r_{u}^2 + \Kuv r_{u}r_{v} + \frac12 \Kvv r_{v}^2 & < \frac{r_v}{M_v} - M_{u}r_{u}, \label{eq:imft2.1} \\
            \Kuv r_{u} + \Kvv r_{v}                                       & < \frac{1}{M_v}, \label{eq:imft2.2}
        \end{align} \label{eq:imft2}
    \end{subequations}
    simultaneously, the implicit map \(g\) from \Cref{thm:imft} is \(\C[2]\) and is valid over \(g:\mathcal{B}(u_0,r_u)\to\mathcal{B}(v_0,r_v)\).
\end{theorem}

\subsection{\LSR{}}
\label{ssec:lsr}

Let \(X \subset \R[n]\) and \(\Lambda \subset \R[m]\) be open sets. Consider the parametrised nonlinear dynamical system \(\dot{x} = \Phi(x,p)\) where \(x \in X\) is the state vector, \(p \in \Lambda\) is the parameter vector, and \(\Phi : X \times \Lambda \to \R[n]\) is at least \(\C[2]\), with a bifurcation of equilibria at some singular point \((x_{\star},p_{\star}) \in X \times \Lambda\), i.e., \(\Phi(x_{\star},p_{\star}) = 0\) and \(\J := \D[x]\Phi(x_{\star},p_{\star})\) is singular with \(q := \dim(\ker\J) > 0\). We write the SVD of \(\J\) as:
\begin{equation}
    \J = \begin{bmatrix}\W \hspace{2mm} \Wbar\end{bmatrix} \begin{bmatrix} \s &\hspace{-2mm} \0[(n-q) \times q] \\ \0[q \times (n-q)] &\hspace{-2mm} \0[q \times q]\end{bmatrix} \begin{bmatrix}\Vbar^{\top} \\ \V^{\top}\end{bmatrix} = \W\s\Vbar^{\top}, \label{eq:SVD}
\end{equation}
where \(\s = \diag(\sigma_{1},\dots,\sigma_{n-q})\) with \(\sigma_{1} \ge \dots \ge \sigma_{n-q} > 0\) is the diagonal matrix of the nonzero singular values of \(\J\), and \(q < n\). Here, \(\V, \Wbar \in \R[n \times q]\) correspond to the null singular values and \(\Vbar, \W \in \R[n \times (n-q)]\) correspond to the nonzero singular values. By the properties of the SVD, the columns of \(\V\) and \(\W\) form orthonormal bases for the null and range spaces of \(\J\), while those of \(\Vbar\) and \(\Wbar\) form orthonormal bases of their respective orthogonal complements.

We define \(\A := \V^{\top}x \in \R[q]\) and \(\B := \Vbar^{\top}x \in \R[n-q]\) to denote the coordinates of \(x \in X\) in the orthonormal bases associated with \(\V\A \in \ker(\J)\) and \(\Vbar\B \in \ker(\J)^{\perp}\) respectively. We define coordinate map \(\Gamma : \R[q] \times \R[n-q] \to \R[n]\)
\begin{equation}
    x = \Gamma(\A,\B) := \V\A + \Vbar\B. \label{def:Gamma}
\end{equation}
Finally, we define the map \(\Phi^{c} : \R[q] \times \R[n-q] \times \R[m] \to \R[n]\) as the composition of dynamics \(\Phi\) with coordinate map \(\Gamma\) as:
\begin{equation*}
    \Phi^{c}(\A,\B,p) := \Phi\bigl(\Gamma(\A,\B),p\bigr) = \Phi\bigl(\V\A+\Vbar\B,p\bigr).
\end{equation*}

\subsubsection*{Application of the \ImFT{}}

The central step of the \LSR{} is the elimination of the \(\B\) coordinates by applying the \ImFT{} to \(f(u,v) \equiv \W^{\top}\Phi^{c}(\A,\B,p)\), with \(u \equiv (\A,p)\) and \(v \equiv \B\). This guarantees the existence of neighbourhoods \(\mathcal{N}\bigl((\A_{\star},p_{\star})\bigr) \subset \R[q] \times \R[m]\) and \(\mathcal{N}(\B_{\star}) \subset \R[n-q]\), together with a unique \(\C[2]\) map \(\varphi : \mathcal{N}\bigl((\A_{\star},p_{\star})\bigr) \to \mathcal{N}(\B_{\star})\) that implicitly determines the \(\B\) coordinates of equilibria near the singular point \((x_{\star},p_{\star})\) through \(\B = \varphi(\A,p)\). Substituting this map and projecting onto \(\range(\J)^{\perp}\) yields the reduced system \(\psi(\A,p) := \Wbar^{\top}\Phi^{c}\bigl(\A,\varphi(\A,p),p\bigr) = \0[q]\), with
\begin{equation}
    \Phi^{c}\bigl(\A,\varphi(\A,p),p\bigr) = \0[n] \iff \psi(\A,p) = \0[q].
    \label{eq:LSRmap}
\end{equation}
The zeros of \(\psi\) define the bifurcation diagram of the reduced system and are locally in one-to-one correspondence with the equilibria of the original \(n\)-dimensional system. The construction is summarised in \Cref{fig:LSR}.

\begin{figure}[pos=ht!]
    \centering
    \begin{tikzpicture}[node distance=1cm]
        \node (Phi) [start] {split the \(n-\dim\) equilibrium \\ equations \(\Phi^{c}(\A,\B,p)=\0[n]\)};
        \coordinate (fork) at ($(Phi.south)+(0,-3.5mm)$);
        \node (perp) [process, below of=fork, xshift=-25mm, yshift=3mm] {set of \(q\) equations \\ \(\Wbar^{\top}\Phi^{c}(\A,\B,p)=\0[q]\)};
        \node (para) [process, below of=fork, xshift=25mm, yshift=3mm] {set of \(n-q\) equations \\ \(\W^{\top}\Phi^{c}(\A,\B,p)=\0[n-q]\)};
        \node (elim) [decision, below of=perp, yshift=-1mm] {eliminate \(\B\) \\ dependence};
        \node (imft) [process, below of=para, yshift=-1mm] {solve locally for \((\A,p) \mapsto \B\) \\ (\ImFT{})};
        \node (project) [stop, below of=elim, xshift=25mm] {reduced \(q-\)dim representation \\ \(\Wbar^{\top}\Phi^{c}\bigl(\A, \varphi(\A,p), p\bigr)=\0[q]\) \\[1ex] of equilibria near bifurcation};

        \draw [line] (Phi.south) -- (fork);
        \draw [arrow] (fork) -| node[midway, align=center, xshift=-8mm, yshift=1ex] {\tiny onto \(\range(\J)\)} (para);
        \draw [arrow] (fork) -| node[midway, align=center, xshift=+9mm, yshift=1ex] {\tiny onto \(\range(\J)^{\perp}\)} (perp);
        \draw [arrow] (perp) -- (elim);
        \draw [arrow] (para) -- (imft);
        \draw [arrow] (imft) -- node[midway, above] {\tiny substitute \(\B=\varphi(\A,p)\)} (elim);
        \draw [arrow] (elim) |- (project);
    \end{tikzpicture}
    \caption{Flowchart depicting the \LSR{} for local bifurcations, as developed in \cite[Chapter VII]{Golubitsky1985}.}
    \label{fig:LSR}
    \vspace{-5mm}
\end{figure}
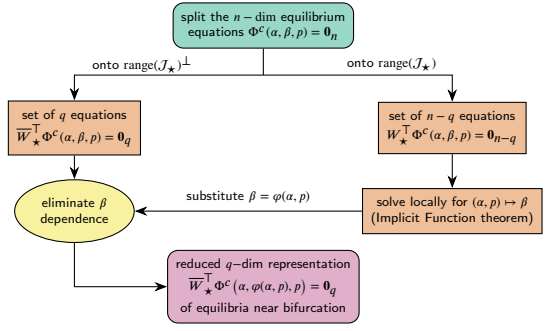

\section{Second-Order Bounds on the Domain of Validity for \LSR{}}
\label{sec:lsr2}

The local graph representation \(\B = \varphi(\A,p)\) obtained through the application of the \ImFT{} is guaranteed to exist on sufficiently small neighbourhoods of \((\A_{\star},p_{\star})\) and \(\B_{\star}\). In this section, we derive explicit lower bounds on the size of these neighbourhoods and hence on the region over which the reduced and full bifurcation diagrams remain topologically equivalent. By specialising \Cref{thm:imft2} to the \LSR{} formulation introduced in \Cref{ssec:lsr}, we obtain computable radii \(\rpr\) and \(\rpp\) such that \(\Bpara[\rpr] \subset \mathcal{N}\bigl((\A_{\star},p_{\star})\bigr)\) and \(\Bperp[\rpp] \subset \mathcal{N}(\B_{\star})\), providing quantitative guarantees on the domain over which the reduction remains valid. The radius \(\rpr\) quantifies robustness along the critical subspace \(\ker(\J) \times \Lambda\), whereas \(\rpp\) quantifies robustness in the transverse directions \(\ker(\J)^{\perp}\).

\begin{theorem}
    \label{thm:lsr2}
    Consider the dynamical system described in \Cref{ssec:lsr}.
    Define the following constants:
    \begin{equation}
        \Mpr := \norm{\W^{\top} \Dp\Phi(x_{\star},p_{\star})}, \quad \Mpp := \norm{\s^{-1}}
        \label{def:lsr2-M}
    \end{equation}
    where \(\s\) is the diagonal matrix of nonzero singular values of \(\J=\Dx\Phi(x_{\star},p_{\star})\). Define the constant matrices
    \begin{align*}
        \chi       & := \begin{bmatrix}\V & \0[n \times m] \\ \0[m \times q] & I_{m}\end{bmatrix} \in \R[(n+m)\times(q+m)], \\
        \bar{\chi} & := \begin{bmatrix} \Vbar \\ \0[m\times(n-q)]\end{bmatrix} \in \R[(n+m)\times(n-q)],
    \end{align*}
    to define the following tensor-valued maps:
    \begin{align}
        \zuu[p] & := \W^{\top} \, \D^2 \Phi(x,p) \bigl[\chi, \chi\bigr]             \in \R[(n-q)\times(q+m)\times(q+m)], \notag                \\
        \zuv[p] & := \W^{\top} \, \D^2 \Phi(x,p) \bigl[\chi, \bar{\chi}\bigr]       \in \R[(n-q)\times(q+m)\times(n-q)], \notag                \\
        \zvv[p] & := \W^{\top} \, \D^2 \Phi(x,p) \bigl[\bar{\chi}, \bar{\chi}\bigr] \in \R[(n-q)\times(n-q)\times(n-q)], \label{def:lsr2-zeta}
    \end{align}
    where \(x=\Gamma(\A,\B)\) and \(\D^2 \Phi(x,p) \in \R[n\times(n+m)\times(n+m)]\) is the Hessian of \(\Phi\) w.r.t. \((x,p) \in \R[n+m]\).
    Given \(\Rpr,\Rpp > 0\) with \(\Bpr[\Rpr] := \Bpara[\Rpr] \subset \R[q] \times \R[m]\) and \(\Bpp[\Rpp] := \Bperp[\Rpp] \subset \R[n-q]\), define constants \(\Kuu, \Kuv, \Kvv\) as the suprema of the respective \(\norm{\zeta}\)s over these balls jointly:
    \begin{align}
        \Kuu & \!:=\! \sup\bigl\{\norm{\zuu[p]} \!:\! (\A,p) \in \Bpr[\Rpr], \B \in \Bpp[\Rpp]\bigr\}, \notag \\
        \Kuv & \!:=\! \sup\bigl\{\norm{\zuv[p]} \!:\! (\A,p) \in \Bpr[\Rpr], \B \in \Bpp[\Rpp]\bigr\}, \notag \\
        \Kvv & \!:=\! \sup\bigl\{\norm{\zvv[p]} \!:\! (\A,p) \in \Bpr[\Rpr], \B \in \Bpp[\Rpp]\bigr\}.
        \label{def:lsr2-K}
    \end{align}
    Then for all \(0 < \rpr < \Rpr\) and \(0 < \rpp < \Rpp\) satisfying
    \begin{subequations}
        \begin{align}
            \frac12\Kuu\rpr^2 + \Kuv\rpr\rpp + \frac12\Kvv\rpp^2 & < \frac{\rpp}{\Mpp} - \Mpr\rpr, \label{eq:lsr2.1} \\
            \Kuv\rpr + \Kvv\rpp                                  & < \frac{1}{\Mpp}, \label{eq:lsr2.2}
        \end{align} \label{eq:lsr2}
    \end{subequations}
    simultaneously, there exists a \(\C[2]\) implicit map \\ \(\varphi : \Bpara[\rpr] \to \Bperp[\rpp]\) that satisfies \eqref{eq:LSRmap}.
\end{theorem}
\begin{proof}
    The result follows by specialising the quantities defined in \Cref{thm:imft2} to the \LSR{} formulation from \Cref{ssec:lsr}. First, the expressions for \(\Mpr\) and \(\Mpp\) are obtained by substituting the Jacobian expressions from \Cref{lem:DF} into the definitions in \eqref{def:imft2-Ms}, following the proofs in \cite[Lemma~3.1, Lemma~3.2]{gupta2025bounds}.
    Then, the \(\zeta\)'s are obtained by expanding the Hessians \(\zuu \equiv \D[(\A,p)]^2 \left[\W^{\top}\Phi^{c}\right]\), \(\zuv \equiv \D[(\A,p)] \Db \left[\W^{\top}\Phi^{c}\right]\) and \(\zvv \equiv \Db^2 \bigl[\W^{\top}\Phi^{c}\bigr]\) in \eqref{def:Ks} as
    \begin{subequations}
        \renewcommand{\arraystretch}{1.25}
        \begin{align}
            \zuu & = \begin{bNiceMatrix}[margin=0.4ex,hvlines]
                         \W^{\top}\,\Da^2\Phi^{c}  & \W^{\top}\,\Da\Dp\Phi^{c} \\ \hline
                         \W^{\top}\,\Dp\Da\Phi^{c} & \W^{\top}\,\Dp^2\Phi^{c}
                     \end{bNiceMatrix}_{i} \label{eq:lsr2-zuu} \\
            \zuv & = \begin{bNiceMatrix}[margin=0.4ex,hvlines]
                         \W^{\top}\,\Da\Db\Phi^{c} \\ \hline
                         \W^{\top}\,\Dp\Db\Phi^{c}
                     \end{bNiceMatrix}_{i} \label{eq:lsr2-zuv}                      \\
            \zvv & = \W^{\top}\,\Db^2 \Phi^{c} \label{eq:lsr2-zvv}
        \end{align} \label{eq:lsr2-zeta}
    \end{subequations}
    Substituting expressions for the Hessians of \(\W^{\top}\Phi^{c}\) from \Cref{lem:D2F,lem:DDF} into the above block forms and grouping terms according to the decompositions induced by \(\chi\) and \(\bar{\chi}\) yields the tensor forms \eqref{def:lsr2-zeta}. Hence, \(\Kuu,\Kuv,\Kvv\) in \eqref{def:lsr2-K} coincide with the corresponding Hessian norm suprema in \eqref{def:Ks}, and the theorem follows.
\end{proof}

In system-specific applications, exact evaluation of \(\Kuu,\Kuv,\Kvv\) may be difficult, as it requires maximising norms of rank-\(3\) tensor-valued maps \(\zuu,\zuv,\zvv\) over prescribed neighbourhoods. When these tensors possess exploitable analytical structure, direct optimisation may instead be replaced by simpler upper bounds. The following corollary shows that any computable upper bounds \(\kuu,\kuv,\kvv\) on these quantities over \(\Bpr[\Rpr]\times\Bpp[\Rpp]\) may be used in place of the exact constants, yielding valid, though potentially more conservative, admissible radii.

\begin{corollary}
    \label{cor:approx}
    Let \(\kuu\ge\Kuu\), \(\kuv\ge\Kuv\), and \(\kvv\ge\Kvv\). Consider a pair
    \(\rhopr,\rhopp>0\) with \(\rhopr<\Rpr\) and \(\rhopp < \Rpp\) satisfying \eqref{eq:lsr2} with the \(K\)s replaced by the \(\tilde{K}\)s:
    \begin{subequations}
        \begin{align}
            \frac12\kuu\rhopr^2 + \kuv\rhopr\rhopp + \frac12\kvv\rhopp^2 & < \frac{\rhopp}{\Mpp} - \Mpr\rhopr, \\
            \kuv\rhopr + \kvv\rhopp < \frac{1}{\Mpp}.
        \end{align} \label{cor:lsr2}
    \end{subequations}
    Then \(\rhopr,\rhopp\) constitute valid bounds for the \LSR{} in the sense of \Cref{thm:lsr2}.
\end{corollary}
\begin{proof}
    Since each \(\tilde{K} \ge K\), it follows that for such a pair of \(\rhopr,\rhopp\) the left-hand sides of \eqref{eq:lsr2} evaluated at \((\rhopr,\rhopp)\) are bounded above by the corresponding left-hand sides of \eqref{cor:lsr2}. Together with \eqref{cor:lsr2}, this shows that \((\rhopr,\rhopp)\) satisfies \eqref{eq:lsr2}. The result therefore follows from \Cref{thm:lsr2}.
\end{proof}

We now characterise the second-order bounds by studying the geometry of the constraints \eqref{eq:lsr2} in the \((\rpr,\rpp)\) plane, yielding an explicit description of the admissible region.

\subsection{Characterisation of the Second-Order Bounds}
\label{ssec:lsr2-geometry}

As illustrated in \Cref{fig:conic}, \eqref{eq:lsr2.1} defines the interior of a conic section in the first quadrant of the \((\rpr, \rpp)\) plane. The admissible portion is always non-empty, intersecting the \(\rpp\) axis for \(0 < \rpp < 2/(\Mpp \Kvv)\) and the \(\rpr \ge 0\) axis only at the origin. Similarly, \eqref{eq:lsr2.2} defines the half-plane below the line \(\Kuv \rpr + \Kvv \rpp = 1/\Mpp\) with positive axis intercepts \(\rpr = 1/(\Mpp\Kuv)\) and \(\rpp = 1/(\Mpp\Kvv)\). The region of admissible \((\rpr,\rpp)\) is naturally the intersection of these sets in the first quadrant.

\begin{figure}[pos=ht!]
    \centering
    \includegraphics[height=48mm]{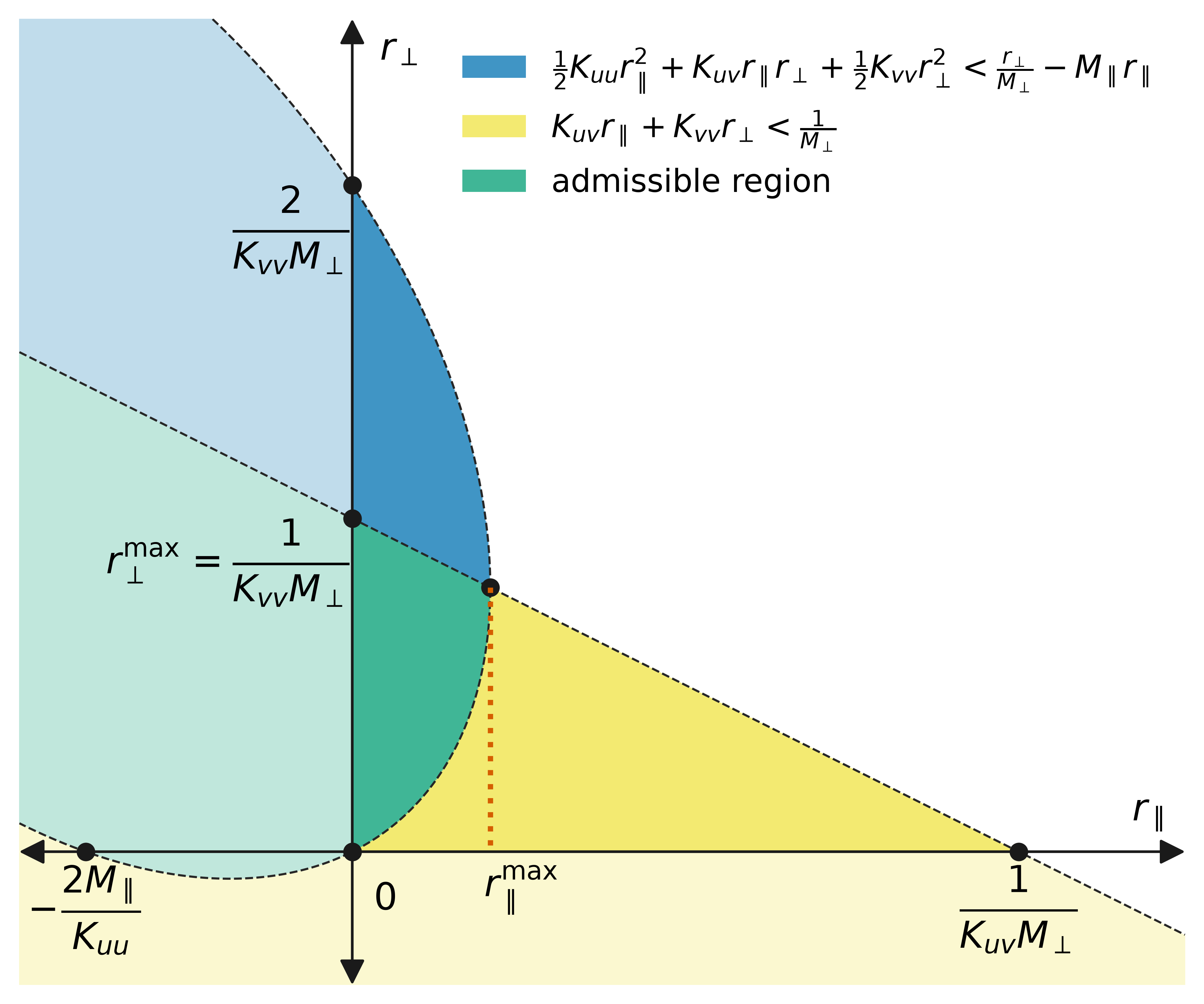}
    \caption{Schematic geometry of the second-order Lyapunov--Schmidt constraints \eqref{eq:lsr2}. Blue and yellow indicate satisfaction of the quadratic \eqref{eq:lsr2.1} and linear \eqref{eq:lsr2.2} inequalities, respectively. Their green overlap represents admissible radius pairs only in the positive quadrant; lighter shading elsewhere illustrates the algebraic continuation. The conic--line intersection determines \(\rpr^{\sup}\), while the line's vertical intercept gives \(\rpp^{\sup}\). The additional restrictions \(\rpr<\Rpr\) and \(\rpp<\Rpp\) are not shown in this figure.}
    \label{fig:conic}
    \vspace{-3mm}
\end{figure}

We now derive explicit expressions for the suprema of radii \(\rpr\) and \(\rpp\) admissible under \eqref{eq:lsr2}. We have already established that \(\rpp^{\sup}=1/(\Kvv\Mpp)\). To determine \(\rpr^{\sup}\), define
\begin{align*}
    F\bigl(\rpr, \rpp\bigr) & = \frac12 \Kuu \rpr^2 + \Kuv \rpr \rpp + \frac12 \Kvv \rpp^2 - \frac{\rpp}{\Mpp} + \Mpr\rpr, \\
    G\bigl(\rpr, \rpp\bigr) & = \Kuv\,\rpr + \Kvv\,\rpp - \frac{1}{\Mpp},
\end{align*}
so that the admissible region is given by \(F,G < 0\).
As depicted in \Cref{fig:conic}, its boundary consists of an arc of \(F=0\) and a segment of \(G=0\). Hence, \(\rpr^{\sup}\) occurs either at a stationary point of \(F=0\) in \(\rpr\), or at an intersection point of \(F=0\) and \(G=0\). Differentiating \(F=0\) implicitly, we get \(\eval{\dv{\rpr}{\rpp}}_{F=0} = -\frac{\pdv{F}/\rpp}{\pdv{F}/{\rpr}}\), so such a stationary point satisfies
\begin{align*}
    \pdv{F}{\rpp} = 0 \implies 0 = \Kuv\rpr+\Kvv\rpp-\frac{1}{\Mpp} = G\bigl(\rpr,\rpp\bigr).
\end{align*}
Hence, such a stationary point is an intersection of \(F=0\) and \(G=0\).
Solving these simultaneously and rationalising yields the following two roots in \((\rpr,\rpp)\):
\begin{align*}
    \rpr & = \frac{1}{\Mpp\left(\Kuv+\Mpr\Mpp\Kvv \pm \sqrt{\Kvv\Theta}\right)},                                      \\
    \rpp & = \frac{\Mpr\Mpp\Kvv \pm \sqrt{\Kvv\Theta}}{\Mpp\Kvv\left(\Kuv+\Mpr\Mpp\Kvv \pm \sqrt{\Kvv\Theta}\right)},
\end{align*}
where \(\Theta:=\Kuu+2\Mpr\Mpp\Kuv+\Mpr^2\Mpp^2\Kvv\).
Since \(\sqrt{\Kvv\Theta} \ge \Mpr\Mpp\Kvv\), the \((-)\) root cannot lie in the first quadrant. Hence, the \((+)\) root is the relevant solution:
\begin{equation}
    \rpr^{\sup} = \frac{1}{\Mpp\left(\Kuv+\Mpr\Mpp\Kvv+\sqrt{\Kvv\Theta}\right)}. \label{eq:rpr-max}
\end{equation}

\subsection{Significance of \(\Rpr\) and \(\Rpp\)}
\label{ssec:lsr2-R}

The parameters \(\Rpr\) and \(\Rpp\) define the neighbourhoods over which the Hessian norm bounds \(\Kuu,\Kuv,\Kvv\) are evaluated. As these constants are defined as suprema over the prescribed neighbourhoods, they are monotonically non-decreasing with respect to \(\Rpr\) and \(\Rpp\). Consequently, enlarging the search region generally yields more conservative estimates of the admissible radii; cf. \eqref{eq:lsr2}. Conversely, the constraints \(\rpr<\Rpr\) and \(\rpp<\Rpp\) imply that choosing \(\Rpr\) or \(\Rpp\) too small may itself limit the certified domain of validity. The selection of \(\Rpr\) and \(\Rpp\) therefore represents a trade-off between the neighbourhood over which the analysis is performed and the size of the resulting guaranteed domain.

\section{Bifurcations in Hopfield-like Networks}
\label{sec:NLD}

In this section, we examine the bounds presented in the preceding section for a class of networked dynamical systems analogous to continuous-time \emph{Hopfield} neural networks. Such systems arise across a broad range of application domains, including recurrent neural network models \cite{hopfield1984neurons,betteti2025input,wilson1972excitatory,sompolinsky1988chaos}, models of opinion dynamics and collective decision-making \cite{franci2015realization,gray2018multiagent,fontan2017multiequilibria,leonard2024fast}, gene regulatory and cellular signalling networks \cite{glass1973logical,cohen1983absolute}, and other interacting systems.
Specifically, we consider dynamical systems defined on a graph \(\mathcal{G}\) of the canonical form
\begin{equation}
    \dot{x} = \Phi(x,p) = -Cx + pAS(x) + b,
    \label{eq:hopfield}
\end{equation}
where \(x \in \R[n]\) is the state, \(p \in \R\) is a scalar parameter, \(b \in \R[n]\) is a constant bias, \(A \in \R[n \times n]\) is the adjacency matrix of \(\mathcal{G}\), and \(C = \mathrm{diag}(c_{1},\dots,c_{n}) \succ 0\). The nonlinear function \(\bigl[S(z)\bigr]_{i}=S_{i}(z_{i})\) applies an elementwise smooth activation function to each input, and its first and second derivatives are given by \(\bigl[\D S(x)\bigr]_{ij}=S'_{i}(x_{i})\delta_{ij}\) and \(\bigl[\D^2S(x)\bigr]_{ijk}=S''_{i}(x_{i})\delta_{ijk}\).

We consider \eqref{eq:hopfield} and denote its Jacobian at the singular point by \(\J[H] := \Dx\Phi(x_{\star},p_{\star}) = -C + p_{\star} A \, \D S(x_{\star})\). We also use the orthonormal basis matrices \(\V\), \(\Vbar\), \(\W\), \(\Wbar\) induced by \(\J[H]\), as well as \(\s\) as described in \eqref{eq:SVD}. Next, we make the following assumptions:
\begin{assumption}
    \label{asmp:hopfield}
    There exists \((x_{\star},p_{\star}) \in \R[n] \times \R[>0]\) such that the following conditions hold:
    \begin{enumerate}
        \item The pair \((x_{\star},p_{\star})\) is an equilibrium of \eqref{eq:hopfield}, i.e.,
              \begin{equation}
                  \Phi(x_{\star}, p_{\star}) = -Cx_{\star} + p_{\star}AS(x_{\star}) + b = 0. \label{hopf:eqbm}
              \end{equation}
        \item The Jacobian \(\J[H]\) satisfies \(\dim{\ker{\J}_{H}} = q\), and \(\operatorname{Re}(\lambda_{j}) \neq 0\) for the remaining nonzero eigenvalues.
    \end{enumerate}
\end{assumption}

We now express the Jacobians and Hessians of \eqref{eq:hopfield} w.r.t. \(x\) and \(p\) for subsequent use:
\begin{alignat}{2}
     & \Dx    \Phi = -C + p A \, \D S(x), \quad &  & \Dp    \Phi = AS(x),        \label{hopf:D1} \\
     & \Dx^2  \Phi = p A \, \D^2S(x),     \quad &  & \Dp^2  \Phi = \0[n],        \label{hopf:D2} \\
     & \Dx\Dp \Phi = A \, \D S(x),        \quad &  & \Dp\Dx \Phi = A \, \D S(x). \label{hopf:DD}
\end{alignat}

In the following sequence of lemmas, we derive expressions for the quantities defined in \Cref{thm:lsr2} required to establish explicit second-order bounds for the validity of the \LSR{} around \((x_{\star},p_{\star})\).

\begin{lemma}
    \label{hopf:M}
    \begin{subequations}
        \begin{align}
            \Mpr & = \norm{\W^{\top} A S(x_{\star})} = \tfrac{1}{\abs{p_{\star}}}\norm{\W^{\top}(C x_{\star} - b)}, \label{hopf:Mpr} \\
            \Mpp & = \norm{\s^{-1}}. \label{hopf:Mpp}
        \end{align}
    \end{subequations}
\end{lemma}
\begin{proof}
    Substituting \(\Dp\Phi(x,p)=AS(x)\) into the definition from \eqref{def:lsr2-M} and using the equilibrium condition \eqref{hopf:eqbm} yields \(AS(x_{\star}) = (Cx_{\star}-b)/p_{\star}\), and consequently \eqref{hopf:Mpr}. Identity \eqref{hopf:Mpp} follows directly from its definition in \eqref{def:lsr2-M}.
\end{proof}

\begin{lemma}
    \label{hopf:zuu}
    \(\zuu \in \R[(n-q)\times(q+1)\times(q+1)]\) can be written as
    \begin{equation*}
        \renewcommand{\arraystretch}{1.25}
        \zuu[p] = \begin{bNiceMatrix}[margin=0.4ex,hvlines]
            p \W^{\top} A \, \D^2S(x) \bigl[\V,\V\bigr] & \W^{\top} A \, \D S(x) \V \\
            \W^{\top} A \, \D S(x) \V                   & \0[n-q]
        \end{bNiceMatrix}_{i}
    \end{equation*}
\end{lemma}
\begin{proof}
    Substituting \eqref{hopf:D2} \(\to\) \eqref{eq:DaaF}, \eqref{hopf:DD} \(\to\) \eqref{eq:DapF}, \eqref{hopf:DD} \(\to\) \eqref{eq:DpaF}, and \eqref{hopf:D2} \(\to\) \eqref{eq:DppF} yields
    \begin{align*}
        \W^{\top}\Da^2\Phi^{c}  & = p\W^{\top}A\,\D^2 S(x) \bigl[\V,\V\bigr] \in \R[(n-q) \times q \times q] \\
        \W^{\top}\Da\Dp\Phi^{c} & = \W^{\top}A\,\D S(x) \V \in \R[(n-q) \times q \times 1]                   \\
        \W^{\top}\Dp\Da\Phi^{c} & = \W^{\top}A\,\D S(x) \V \in \R[(n-q) \times 1 \times q]                   \\
        \W^{\top}\Dp^2\Phi^{c}  & = \W^{\top}\0[n] = \0[n-q] \in \R[(n-q)\times1\times1]
    \end{align*}
    Collecting these into the blocks of \(\zuu[p]\) expressed in \eqref{eq:lsr2-zuu} yields the desired result.
\end{proof}

\begin{lemma}
    \label{hopf:zuv}
    \(\zuv \in \R[(n-q)\times(q+1)\times(n-q)]\) can be written as
    \begin{equation*}
        \renewcommand{\arraystretch}{1.25}
        \zuv[p] = \begin{bNiceMatrix}[margin=0.4ex,hvlines]
            p\W^{\top}A\,\D^2S(x)[\V,\Vbar] \\
            \W^{\top}A\,\D S(x) \Vbar
        \end{bNiceMatrix}_{i}
    \end{equation*}
\end{lemma}
\begin{proof}
    Substituting \eqref{hopf:D2} \(\to\) \eqref{eq:DabF} and \eqref{hopf:DD} \(\to\) \eqref{eq:DpbF},
    \begin{align*}
        \W^{\top}\Da\Db\Phi^{c} & = p\W^{\top}A\,\D^2S(x) \bigl[\V,\Vbar\bigr] \in \R[(n-q)\times q \times (n-q)] \\
        \W^{\top}\Dp\Db\Phi^{c} & = \W^{\top}A\,\D S(x) \Vbar \in \R[(n-q)\times 1 \times (n-q)]
    \end{align*}
    Collecting these into the blocks of \(\zuv[p]\) expressed in \eqref{eq:lsr2-zuv} yields the desired result.
\end{proof}

\begin{lemma}
    \label{hopf:zvv}
    \(\zvv \in \R[(n-q)\times(n-q)\times(n-q)]\) can be written as
    \begin{equation*}
        \zvv[p] = p \W^{\top} A \, \D^2S(x)[\Vbar,\Vbar]
    \end{equation*}
\end{lemma}
\begin{proof}
    Substituting \eqref{hopf:D2} \(\to\) \eqref{eq:DbbF} yields
    \begin{equation*}
        \W^{\top}\,\Db^2\Phi^{c}(\A,\B,p) = \W^{\top} \bigl(pA\,\D^2S(x)\bigr) [\Vbar, \Vbar]
    \end{equation*}
    which, together with \eqref{eq:lsr2-zvv} yields the desired expression.
\end{proof}

The expressions derived above hold for general Hopfield networks. We now specialise to the regular graph case, where the derived expressions for \(\zuu,\zuv,\zvv\) permit explicit upper bounds on \(\Kuu,\Kuv,\Kvv\), and hence analytical second-order validity bounds under \Cref{cor:approx}.

\subsection*{Evaluation for Regular Graphs}

We specialise \eqref{eq:hopfield} to the unbiased case \((b \equiv \0[n])\) over regular graphs by taking \(C \equiv dI_{n}\) for some \(d>0\), \(p \equiv u\), \(S \equiv \tanh\), and \(A\) to be the binary adjacency matrix of a connected, undirected, simple \(k-\)regular graph:
\begin{equation}
    \dot{x} = \Phi(x,u) = -dx + uA\tanh(x). \label{eq:regular}
\end{equation}
This system undergoes a supercritical pitchfork bifurcation at \((x_{\star},u_{\star})=\left(\0[n], d/k\right)\) at which branches of equilibria emerge. By \eqref{def:Gamma}, \((\A_{\star},\B_{\star})=(0,\0[n-1])\).

\begin{figure}[ht!]
    \centering
    \includegraphics[width=.55\linewidth]{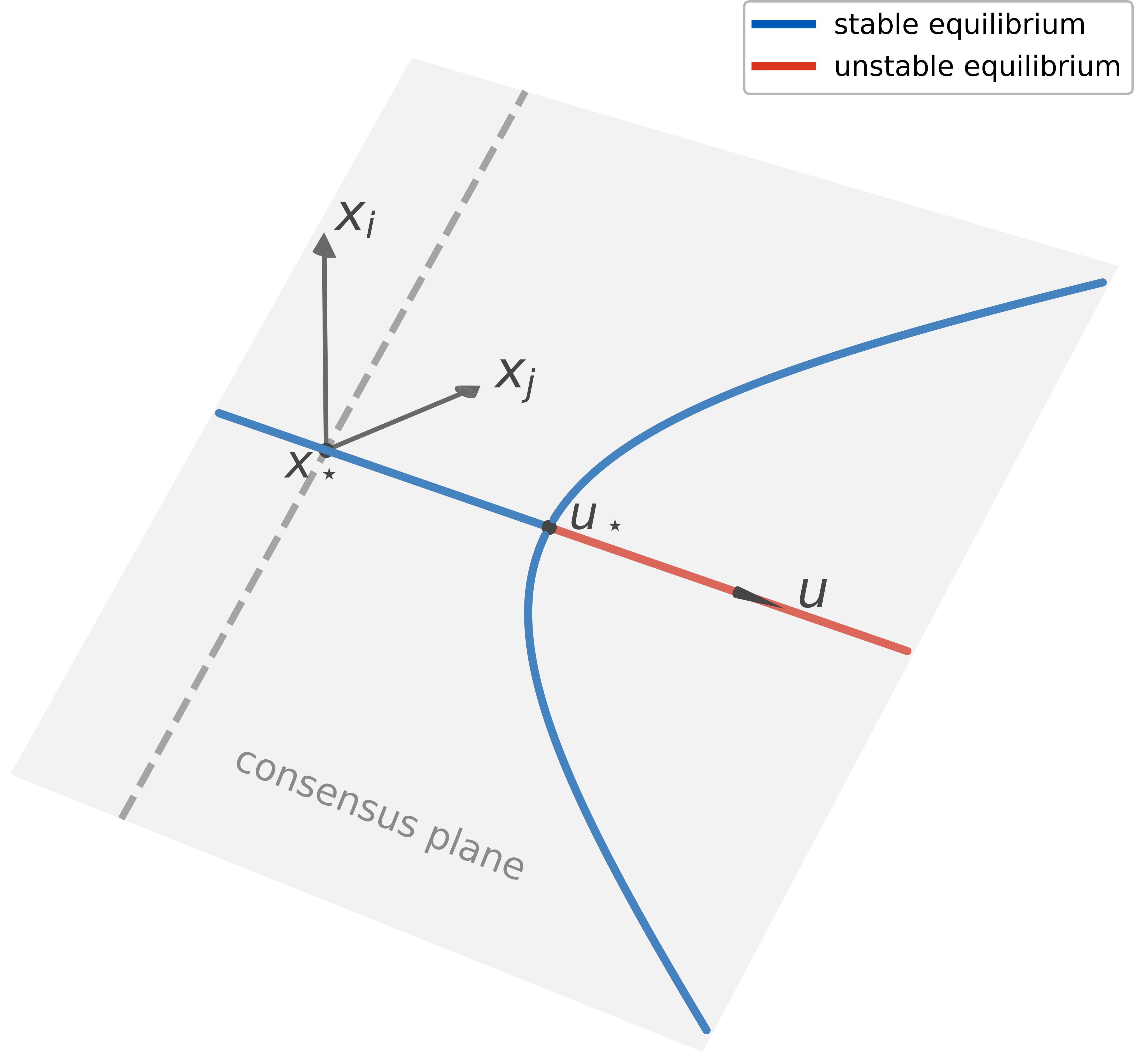}
    \caption{Primary pitchfork bifurcation in Hopfield network \eqref{eq:regular} at \((x_{\star},u_{\star})=\left(\0[n], d/k\right)\); branches of equilibria emerge from the primary pitchfork within the consensus plane, which is tangent to the consensus (centre) manifold at the bifurcation point.}
    \vspace{-3mm}
\end{figure}

The Jacobian \(\J := \Dx\Phi(x_{\star},u_{\star})\) simplifies to
\begin{equation}
    \J = -dI_{n}+u_{\star}A\,\mathrm{diag}\bigl[\sech^2(x_{\star})\bigr] = \tfrac{d}{k}(A-kI_{n}). \label{eq:J}
\end{equation}

We first characterise the orthonormal basis matrix \(\V\) of \(\ker(\J)\). By the Perron--Frobenius theorem \cite[Theorem~8.4.4]{horn2012matrix}, \(A\) has a simple largest eigenvalue \(k\), with Perron eigenvector \(\1[n]\). Consequently, \(\ker(\J)\) is one-dimensional (\(q=1\)), spanned by \(\V = \1[n]/\sqrt{n}\). Since \(\J\) is symmetric, \(\W = \Vbar\) and \(\V = \Wbar\). We also define \(\U \in \R[(n-1)\times(n-1)]\) as the diagonal matrix of all non-Perron eigenvalues \((\neq k)\) of \(A\), and let \(\W \in \R[n\times(n-1)]\) contain the corresponding orthonormal eigenvectors as columns so that \(A\W = \W\U\). With these simplifications, we now adapt the expressions for the \(\zeta\) tensors from \Cref{hopf:zuu,hopf:zuv,hopf:zvv} to the regular graph case, and characterise their Frobenius norms.

We employ the Frobenius norm rather than the spectral norm, as it permits explicit evaluation of the norms of the rank-3 \(\zeta\) tensors and facilitates maximisation over the prescribed neighbourhoods. By contrast, the spectral norm does not readily admit analogous analytical or numerical treatment. Even with the Frobenius norm, however, exact maximisation remains analytically intractable for general regular graphs. To obtain explicit bounds, we instead employ \Cref{cor:approx}, using upper bounds on the derivatives of the activation \(S=\tanh(\cdot)\) to derive tractable upper bounds on the \(\norm{\zeta}\)s depending only on \(u\), the network parameters \(n\) and \(k\), and the constants \(m_1=1\) and \(m_2=4/(3\sqrt{3})\), which bound \(S'=\sech^2(\cdot)\) and \(S''=-2\tanh(\cdot)\sech^2(\cdot)\), respectively. Before establishing bounds on the \(\norm{\zeta}\)s, we record the following useful lemmas.

\begin{lemma}
    \label{lem:WAW}
    \small \(\norm{\W^{\top}A\W}^2_F = \norm{\W^{\top}A}^2_F = \norm{\U}^2_F = k(n-k)\).
\end{lemma}
\begin{proof}
    By construction, \(\W^{\top}A\W = \U\). Moreover, we have
    \(
    \norm{\W^\top A}_F^2
    =\norm{\U\W^{\top}}_F^2
    =\tr(\U \W^{\top} \W \U)
    =\tr(\U^2)
    =\norm{\U}_F^2
    \), since \(\W^{\top}\W=I_{n-1}\).
    Finally, \(\norm{\U}_F^2 = \sum_{i=2}^n \lambda_i(A)^2 = \norm{A}_F^2 - k^2 = nk-k^2 = k(n-k)\), where \(\norm{A}_F^2 = \sum A_{ij}^2 = nk\) since each node has degree \(k\).
\end{proof}

\begin{lemma}
    \label{lem:WADW}
    Let \(y \in \R[n]\) with \(\abs{y_i} \le m\) for all \(i\). Then,
    \begin{equation*}
        \max_{y\in[-m,m]^n} \norm{\W^{\top} A \diag(y) \W} = m\norm{\U},
    \end{equation*}
    for both the spectral and Frobenius norms.
\end{lemma}
\begin{proof}
    Using \(\W^{\top} A = \U \W^{\top}\) together with the submultiplicativity inequality \(\norm{AB}_F \le \norm{A}_F\norm{B}_2\), we have
    \begin{align*}
        \norm{\W^{\top} A \diag(y) \W}_F & = \norm{\U \W^{\top} \diag(y) \W}_F                 \\
                                         & \le \norm{\U}_F \norm{\W^{\top} \diag(y) \W}_2      \\
                                         & \le \norm{\U}_F \norm{\diag(y)}_2 \le m\norm{\U}_F.
    \end{align*}
    The above also holds for the spectral norm by its submultiplicativity. Since equality is attained at \(y = m\mathbf{1}_n\), the bound is tight and coincides with the maximum value of \(\norm{\W^{\top} A \diag(y) \W}\) over the hypercube \([-m,m]^n\).
\end{proof}

\subsubsection*{Bounding the $\norm{\zeta}$s}

We now specialise the expressions for the \(\zeta\) tensors from \Cref{hopf:zuu,hopf:zuv,hopf:zvv} to the connected regular graph setting \eqref{eq:regular}, and derive bounds on their Frobenius norms that are independent of \(\A\) and \(\B\) and depend only on \(u\).

Starting with \(\zuu\), we build on \Cref{hopf:zuu} to get
\begin{equation*}
    \renewcommand{\arraystretch}{1.25}
    \zuu[u] = \begin{bNiceMatrix}[margin=0.4ex,hvlines]
        \frac{u}{n}\W^{\top}AS''(x)       & \frac{1}{\sqrt{n}}\W^{\top}AS'(x) \\
        \frac{1}{\sqrt{n}}\W^{\top}AS'(x) & \0[n-1]
    \end{bNiceMatrix}_{i}.
\end{equation*}
Under the Frobenius norm,
\begin{align}
    \norm{\zuu[u]}_F^2 & = \tfrac{u^2}{n^2} \norm{\W^{\top} A S''(x)}_F^2 + \tfrac{2}{n} \norm{\W^{\top} A S'(x)}_F^2               \notag \\
                       & \le \norm{\W^{\top} A}_F^2 \left[\tfrac{u^2}{n^2} \norm{S''(x)}_F^2 + \tfrac{2}{n} \norm{S'(x)}_F^2\right] \notag \\
                       & \le \bigl[k(n-k)\bigr]\left[\tfrac{u^2}{n^2}\bigl(n m_2^2\bigr) + \tfrac{2}{n}\bigl(n m_1^2\bigr)\right]   \notag \\
                       & = \bigl[k(n-k)\bigr]\left[\tfrac{u^2}{n} m_2^2 + 2 m_1^2\right], \label{hopf:zuu-regular}
\end{align}
Next, for \(\zuv\), building on \Cref{hopf:zuv} yields
\begin{equation}
    \renewcommand{\arraystretch}{1.25}
    \zuv[u] = \begin{bNiceMatrix}[margin=0.4ex,hvlines]
        \frac{u}{\sqrt{n}} \W^{\top} A \diag\bigl(S''(x)\bigr) \Vbar \\
        \W^{\top} A \diag\bigl(S'(x)\bigr) \Vbar
    \end{bNiceMatrix}_{i} \label{eq:zuv_hopf_reg}
\end{equation}
Under the Frobenius norm,
\begin{align*}
    \norm{\zuv[u]}_F^2 & = \tfrac{u^2}{n}\norm{\W^{\top} A \diag\bigl(S''(x)\bigr) \Vbar}_F^2 + \\ & \hspace{20mm} \norm{\W^{\top} A \diag\bigl(S'(x)\bigr) \Vbar}_F^2.
\end{align*}
Applying \Cref{lem:WADW} to the above expression yields
\begin{align}
    \norm{\zuv[u]}_F^2 & \le \norm{\U}_F^2 \left[\tfrac{u^2}{n} m_2^2 + m_1^2\right] \notag                      \\
                       & = \bigl[k(n-k)\bigr]\left[\tfrac{u^2}{n} m_2^2 + m_1^2\right], \label{hopf:zuv-regular}
\end{align}
Finally, for \(\zvv\), using \Cref{hopf:zvv} in index notation,
\begin{align*}
    \bigl[\zvv[u]\bigr]_{ijk} & = u \W[\mu i] A_{\mu\nu} S''_{\nu}(x_{\nu}) \Vbar[\nu j] \Vbar[\nu k] \\
                              & = u \U[ii] S''_{\nu}(x_{\nu}) \W[\nu i] \W[\nu j] \W[\nu k],
\end{align*}
for \(i,j,k \in \{1,\dots,n-1\}\) and \(\mu,\nu \in \{1,\dots,n\}\), since \(\Vbar = \W\) and \(\W[\mu i]A_{\mu\nu} = (\W^{\top} A)_{i\nu} = (\U\W^{\top})_{i\nu} = \U[ii]\W[\nu i]\).
Expanding the Frobenius norm gives
\begin{equation*}
    \norm{\zvv[u]}_F^2 \! = \! u^{2}\!\sum_{i,j,k}\!\U[ii]^{2}\!\biggl(\sum_{\nu=1}^{n} S''_{\nu}(x_{\nu}) \W[\nu i]\W[\nu j]\W[\nu k]\!\biggr)^{\!2}.
\end{equation*}
In \Cref{app:Q}, we show that the summation \\ \(\sum_{j,k} \left(\sum_{\nu=1}^{n} S''(x_\nu) \W[\nu i]\W[\nu j]\W[\nu k]\right)^{2} \le m_2^2(1-1/n)\). \\
Substituting this into the above yields
\begin{align}
    \norm{\zvv[u]}_F^2 & \le (u m_2)^2 \left(1-\frac1n\right) \sum_{i=1}^{n-1}\U[ii]^2 \notag \\
                       & = (u m_2)^2 k(n-k)\left(1-1/n\right), \label{hopf:zvv-regular}
\end{align}
since \(\sum_{i} \U[ii]^2 = \norm{\U}_F^2 = k(n-k)\).

We are now ready to state explicit second-order bounds for the primary bifurcation in \eqref{eq:regular} at \((x_{\star},u_{\star})=(\0[n],d/k)\).

\begin{theorem}
    \label{thm:lsr2-regular}
    Consider the unbiased Hopfield-like dynamics in \eqref{eq:regular}.
    For fixed \(\Rpr > 0\), define the constants
    \begin{align*}
        \kuu & = \sqrt{k(n-k)}\sqrt{\tfrac{(d/k+\Rpr)^2}{n} m_2^2 + 2 m_1^2}, \\
        \kuv & = \sqrt{k(n-k)}\sqrt{\tfrac{(d/k+\Rpr)^2}{n} m_2^2 + m_1^2},   \\
        \kvv & = m_2(d/k+\Rpr)\sqrt{k(n-k)(1-1/n)}.
    \end{align*}
    Then for all \(0 < \rhopr < \Rpr\) and \(0 < \rhopp\) satisfying
    \begin{align*}
        \frac12\kuu\rhopr^2 + \kuv\rhopr\rhopp + \frac12\kvv\rhopp^2 & < \frac{\rhopp}{\tfrac{k}{d}\sqrt{\sum_{i=2}^{n}\tfrac{1}{(k-\lambda_i)^{2}}}}, \\
        \kuv\rhopr + \kvv\rhopp                                      & < \frac{1}{\tfrac{k}{d}\sqrt{\sum_{i=2}^{n}\tfrac{1}{(k-\lambda_i)^{2}}}},
    \end{align*}
    where \(k = \lambda_1 > \lambda_2 \ge \dots \ge \lambda_n\) are the ordered eigenvalues of \(A\),
    there exists a smooth implicit map \(\varphi : \mathcal{B}\bigl((0,d/k), \rhopr\bigr) \to \mathcal{B}(\0[n-1], \rhopp)\) that satisfies \eqref{eq:LSRmap} for the dynamics in \eqref{eq:regular}.
\end{theorem}
\begin{proof}
    The result follows by verifying the constants appearing in \Cref{thm:lsr2}. First, \(\Mpr = \norm{\W^{\top}A\tanh(\0[n])}_F = 0\).
    Next, we evaluate \(\Mpp\) by computing \(\s\) using \eqref{eq:J} as
    \begin{align*}
        \s & = \W^{\top}\J\Vbar = \tfrac{d}{k}\W^{\top}(A-kI_{n})\W = \tfrac{d}{k}\bigl(\U-kI_{n-1}\bigr) \\
           & \implies \Mpp = \norm{\s^{-1}}_F = \frac{k}{d}\sqrt{\sum_{i=2}^{n}\frac{1}{(k-\lambda_i)^2}}
    \end{align*}
    Finally, applying \Cref{cor:approx} to the upper bounds on \(\norm{\zeta}\)s from \eqref{hopf:zuu-regular}, \eqref{hopf:zuv-regular}, and \eqref{hopf:zvv-regular}, together with \(\abs{u} < d/k + \Rpr\) yields the required expressions for \(\kuu,\kuv,\kvv\), and the claimed bounds follow immediately.
\end{proof}

\section{Comparison with First-Order Bounds}
\label{sec:1v2}

Having derived the second-order bounds, we now compare them with the first-order bounds of \cite{gupta2025bounds}, recalled below in \Cref{thm:lsr1,thm:lsr1-regular}. We begin in \Cref{ssec:conservative} with a comparison of the maximal admissible radii for Hopfield dynamics on complete graphs, providing preliminary intuition for the relative conservativeness of the second-order bounds. In \Cref{ssec:general}, a general degeneracy condition extends the inclusion result to all connected regular graphs. Finally, we show via a simple counterexample that this inclusion cannot be asserted universally for arbitrary dynamical systems.

\begin{theorem}
    \label{thm:lsr1}
    \cite[Theorem~3.3]{gupta2025bounds}
    Building upon the definitions from \Cref{ssec:lsr}, we define the constants
    \begin{equation}
        \Mpr = \norm{\W^{\top} \Dp\Phi(x_{\star},p_{\star})}, \quad \Mpp = \norm{\s^{-1}},
        \label{def:lsr1-M}
    \end{equation}
    and maps \(\xi_{\parallel} : \R[q] \times \R[m] \to \R[(n-q)\times(q+m)]\) and \(\xi_{\perp} : \R[q] \times \R[n-q] \times \R[m] \to \R[(n-q)\times(n-q)]\) as
    \begin{subequations}
        \begin{align}
            \xipr{p} & := \W^{\top}\begin{bmatrix} \xi_{\parallel\A}(\A,p) & \xi_{\parallel p}(\A,p) \end{bmatrix}, \label{def:xi_para} \\
            \xipp{p} & := \W^{\top}\Dx\Phi\bigl(\Gamma(\A,\B),p\bigr)\Vbar - \s, \label{def:xi_perp}
        \end{align} \label{def:lsr1-xi}
    \end{subequations}
    where \(\xi_{\parallel\A}(\A,p) = \Dx\Phi\bigl(\Gamma(\A,\B_{\star}),p\bigr)\V\) and \(\xi_{\parallel p}(\A,p) = \Dp\Phi\bigl(\Gamma(\A,\B_{\star}),p\bigr)-\Dp\Phi(x_{\star},p_{\star})\). We introduce shorthand notations \(\Bpr[\rpr] := \Bpara[\rpr] \subset \R[q] \times \R[m]\) and \(\Bpp[\rpp] := \Bperp[\rpp] \subset \R[n-q]\) to denote open balls centred at the singular point. Finally, for \(\rpr,\rpp > 0\), define
    \begin{subequations}
        \begin{align}
            \Lpr & := \sup_{(\A,p) \in \Bpr[\rpr]} \norm{\xipr{p}}, \label{def:L_para} \\
            \Lpp & := \sup_{\substack{(\A,p) \in \Bpr[\rpr]                            \\ \B \in \Bpp[\rpp]}} \norm{\xipp{p}}. \label{def:L_perp}
        \end{align} \label{def:lsr1-L}
    \end{subequations}
    Then, for all \(\rpr,\rpp>0\) satisfying
    \begin{equation}
        \Lpr\rpr + \Lpp\rpp < \frac{\rpp}{\Mpp} - \Mpr\rpr \label{eq:lsr1}
    \end{equation}
    there exists a smooth implicit map \\ \(\varphi : \Bpara[\rpr] \to \Bperp[\rpp]\) that satisfies \eqref{eq:LSRmap}.
\end{theorem}

\begin{theorem}
    \label{thm:lsr1-regular}
    Consider \eqref{eq:regular} over a connected regular graph with \(n\) nodes of degree \(k\). Sufficient first-order bounds for the \LSR{} from \Cref{thm:lsr1} around the pitchfork bifurcation at \((x_{\star},u_{\star})=(\0[n], d/k)\) are given by
    \small
    \begin{equation*}
        \norm{\U}\norm{\s^{-1}}\!\left(\rpr + \max\left\{\!\tfrac{d}{k}-\rpr,0\!\right\}\tanh^2\left(\tfrac{\rpr}{\sqrt{n}} + \sqrt{\tfrac{n-1}{n}}\rpp\right)\!\!\right)<1,
    \end{equation*}
    \normalsize
    where \(\rpr>0\) and \(\rpp>0\) are the radii along \(\ker(\J)\times\R\) and \(\ker(\J)^\perp\), respectively. Whenever the graph satisfies \(\frac{d}{k}\norm{\s^{-1}}\norm{\U} \le 1\), these bounds simplify to
    \begin{equation}
        0<\rpr<\frac{1}{\norm{\U}\norm{\s^{-1}}} \quad\text{and}\quad \rpp>0. \label{eq:decouple}
    \end{equation}
\end{theorem}
\begin{proof}
    From \cite[Section~V]{gupta2025bounds}, we have \(\Mpr=\Lpr=0\) for both the spectral and Frobenius norms, and
    \begin{equation*}
        \norm{\xipp{u}} = \norm{\W^{\top}A\diag\left(u\sech^2\bigl(\Gamma(\A,\B)\bigr) - \tfrac{d}{k}\right)\W}.
    \end{equation*}
    Using \(\abs{u-d/k}<\rpr\) (arising from the domain's radius constraint) and \(\sech^2(\cdot) = 1 - \tanh^2(\cdot)\), we first obtain
    \begin{align*}
        \abs{u\sech^2(\Gamma_i)-\tfrac{d}{k}} &\le \abs{u-\tfrac{d}{k}}\sech^2(\Gamma_i) +\tfrac{d}{k}\tanh^2(\Gamma_i) \\
        &\le \rpr+\left(\tfrac{d}{k}-\rpr\right) \tanh^2(\Gamma_i) \\
        &\le \rpr+\max\left\{\tfrac{d}{k}-\rpr,0\right\} \tanh^2(\Gamma_i).
    \end{align*}
    Applying Cauchy--Schwarz elementwise to \(\Gamma_i\) from \eqref{def:Gamma},
    \begin{equation*}
        \abs{\Gamma_i(\A,\B)} \le \tfrac{\abs{\A}}{\sqrt n} + \sqrt{\tfrac{n-1}{n}}\norm{\B}_2 \le \tfrac{\rpr}{\sqrt n} +\sqrt{\tfrac{n-1}{n}}\rpp.
    \end{equation*}
    Combining these with \Cref{lem:WADW} and substituting the above into \(\Mpp\Lpp<1\) yields the first stated condition. Finally, if \(\frac{d}{k}\norm{\U}\norm{\s^{-1}} \le 1\), this condition holds automatically for \(0 < \rpr < d/k\), while for \(\rpr \ge d/k\) it reduces to \(\rpr\norm{\U}\norm{\s^{-1}}<1\), which completes the proof.
\end{proof}
\begin{remark}[Choice of norm]
    The first-order bounds in \Cref{thm:lsr1-regular} have the same analytical form for the spectral and Frobenius norms, but their values and the satisfaction of \(\frac{d}{k}\norm{\s^{-1}}\norm{\U} \le 1\) depend on the choice of norm and the graph spectrum. An alternative proof based on monotonicity arguments can be used to establish the bound \eqref{eq:decouple} for all connected regular graphs; we choose to present the more conservative condition here to emphasise the relationship between the first-order bounds of \Cref{thm:lsr1-regular} and the second-order bounds presented in \Cref{thm:lsr2-regular}. Our numerical experiments use the Frobenius norm.
\end{remark}

\subsection{Preliminary Comparison for Complete Graphs}
\label{ssec:conservative}

Before addressing the general case, it is instructive to consider a setting in which the first-order bound takes the simple, decoupled form of \eqref{eq:decouple}. By \Cref{thm:lsr1-regular}, this simplification is guaranteed whenever \(\tfrac{d}{k}\norm{\s^{-1}}\norm{\U} \le 1\), a condition satisfied by every complete graph under both the Frobenius and spectral norms. Restricting attention to this case for now lets us compare the first and second-order maximal radii directly in closed form, giving preliminary insight into the relative conservativeness of the two bounds.

Since the first-order validity region is unbounded in \(\rpp\), whereas the second-order validity region satisfies \(\rpp^{\sup} = 1/(\Mpp \Kvv)\), it remains only to compare the admissible ranges of \(\rpr\). Specifically, it suffices to show that the maximal admissible \(\rpr\) under the second-order bounds does not exceed \(1/\Bigl(\norm{\U}\norm{\s^{-1}}\Bigr)\).
Substituting \(\Mpr=0\) into \eqref{eq:rpr-max} yields
\begin{equation*}
    \rpr^{\sup} = \frac{1}{\Mpp(\Kuv+\sqrt{\Kuu\Kvv})}.
\end{equation*}
Since \(\Mpp=\norm{\s^{-1}}\), it therefore suffices to establish
\begin{align}
    \frac{1}{\Mpp(\Kuv+\sqrt{\Kuu\Kvv})} \le \frac{1}{\Mpp\norm{\U}}
    \notag \\ \iff
    \boxed{\Kuv+\sqrt{\Kuu\Kvv} \ge \norm{\U}} \label{eq:final}
\end{align}
By the definition of \(\Kuv\) in \eqref{def:lsr2-K}, \(\Kuv \ge \norm{\zuv\bigl(0, \0[n-1], u_{\star}\bigr)}\). Evaluating \(\zuv(0, \0[n-1], u_{\star})\) from \eqref{eq:zuv_hopf_reg}, we obtain
\begin{equation*}
    \renewcommand{\arraystretch}{1.25}
    \zuv\bigl(0, \0[n-1], u_{\star}\bigr) = \begin{bNiceMatrix}[margin=0.4ex,hvlines]
        \0[(n-1)\times(n-1)] \\ \W^{\top} A\W
    \end{bNiceMatrix}_{i}
\end{equation*}
since \(S''\bigl(\0[n]\bigr) = \0[n]\) and \(S'\bigl(\0[n]\bigr)=\1[n]\) for \(S \equiv \tanh\). Hence,
\begin{equation*}
    \norm{\zuv\bigl(0, \0[n-1], u_{\star}\bigr)} = \norm{\W^{\top}A\W} = \norm{\U}
\end{equation*}
for both the spectral and Frobenius norms. Therefore, \(\Kuv \ge \norm{\U}\). Since \(\Kuu, \Kvv \ge 0\), we have \(\Kuv + \sqrt{\Kuu\Kvv} \ge \Kuv \ge \norm{\U}\), which establishes \eqref{eq:final}. Consequently, for complete graphs, the second-order admissible region is contained in the first-order strip \eqref{eq:decouple}. We show next that this conclusion extends to every connected regular graph.

\subsection{A General Comment on the Conservativeness}
\label{ssec:general}

The comparison of \Cref{ssec:conservative} was carried out only for complete graphs, where the first-order bound takes the simple, decoupled form of \Cref{eq:decouple}. To extend the inclusion result to the remaining connected regular graphs, we instead compare the general first and second-order conditions \eqref{eq:lsr1} and \eqref{eq:lsr2} directly. This comparison identifies a degeneracy condition that holds for Hopfield dynamics on every connected regular graph and also applies beyond this network class. Finally, we show via a simple counterexample that outside this degenerate regime, no such inclusion can be expected to hold universally for arbitrary systems.

The quantities \(\Lpr,\Lpp\) defined in the first-order bounds in \eqref{def:lsr1-L} and \(\Kuu,\Kuv,\Kvv\) of the second-order bounds in \eqref{def:lsr2-K} are not independent: since \(\xipr{p},\xipp{p}\) are themselves increments of a Jacobian block of \(\Phi\) from its value at \((\A_{\star},\B_{\star},p_{\star})\), the mean value theorem bounds them directly by the corresponding Hessian suprema, in the spirit of the argument relating the two \ImFT{} formulations in \cite{jindal2024estimates}.

\begin{theorem}
    \label{thm:LK}
    For all \(0<\rpr<\Rpr\) and \(0<\rpp<\Rpp\),
    \begin{equation}
        \Lpr \le \Kuu\rpr, \qquad \Lpp \le \Kuv\rpr + \Kvv\rpp. \label{eq:LK}
    \end{equation}
\end{theorem}
\begin{proof}
    Write \(\eta(\A,p):=\begin{bmatrix}\Da\phi(\A,\B_{\star},p) & \Dp\phi(\A,\B_{\star},p)\end{bmatrix}\), where \(\phi:=\W^{\top}\Phi^c\). Since \(\Da\phi(\A_{\star},\B_{\star},p_{\star})=\W^{\top}\J\V=\0\), comparing with \eqref{def:xi_para} gives \(\xipr{p}=\eta(\A,p)-\eta(\A_{\star},p_{\star})\). Integrating along
    \begin{equation*}
        \bigl(\A(t),p(t)\bigr) := (\A_{\star},p_{\star})+t\bigl[(\A,p)-(\A_{\star},p_{\star})\bigr] \quad t\in[0,1],
    \end{equation*}
    the fundamental theorem of calculus and \eqref{eq:lsr2-zuu} give
    \begin{equation*}
        \xipr{p}=\int_0^1 \zuu\bigl(\A(t),\B_{\star}, p(t)\bigr)\bigl[(\A,p)-(\A_{\star},p_{\star})\bigr]\dd{t},
    \end{equation*}
    so \(\norm{\xipr{p}} \le \Kuu \norm{(\A,p)-(\A_{\star},p_{\star})} \le \Kuu\rpr\) for \((\A,p)\in\Bpr[\rpr]\); taking the supremum gives the first bound. A similar construction using the fundamental theorem of calculus and \eqref{eq:lsr2-zuv} shows that \(\Lpp \le \Kuv\rpr + \Kvv\rpp\).
\end{proof}

In the following corollary, we use \Cref{thm:LK} to establish a general sufficient condition for inclusion based on the \emph{degeneracy} \(\Mpr=\Lpr=0\). This generalises the preliminary complete-graph comparison in \Cref{ssec:conservative} to all connected regular graphs, and further applies to other systems exhibiting the same degeneracy.

\begin{corollary}
    \label{cor:degenerate}
    If \(\Mpr=\Lpr=0\), then every \((\rpr,\rpp)\) satisfying \eqref{eq:lsr2} also satisfies \eqref{eq:lsr1}; the second-order domain of validity is contained in the first-order domain.
\end{corollary}
\begin{proof}
    With \(\Mpr=0\) and \(\Lpr=0\), \eqref{eq:lsr1} reduces to \(\Lpp<1/\Mpp\). Combining \eqref{eq:lsr2.2} with \Cref{thm:LK}, \(\Lpp \le \Kuv\rpr+\Kvv\rpp < 1/\Mpp\), which is precisely this condition.
\end{proof}

For Hopfield dynamics \eqref{eq:regular} on any connected regular graph, \(\Mpr=\Lpr=0\), as established in the proof of \Cref{thm:lsr1-regular}. Hence, \Cref{cor:degenerate} shows that every radius pair certified by \Cref{thm:lsr2} is also certified by \Cref{thm:lsr1}, irrespective of whether the graph satisfies \(\frac{d}{k}\norm{\s^{-1}}\norm{\U} \le 1\).

Away from this degenerate regime, however, \emph{neither} bound dominates the other in general; which one is tighter depends on \(\Phi\) itself. Consider \(f(u,v)=v+u^{p}\) at \((u_0,v_0)=(0,0)\) with \(p\in\{2,3,\dots\}\). Fix arbitrarily large \(R_u,R_v>0\), and compare the two certificates for \(0<r_u<R_u\) and \(0<r_v<R_v\). Here, \(M_u=\norm{D_uf(0,0)}=0\), \(M_v=1\), and \(D_vf\equiv1\), so \(L_v\equiv0\); only the \(u\)-direction contributes to the radius conditions, with \(L_u(r_u)=\sup_{\abs{u}<r_u}\abs{pu^{p-1}}=pr_u^{p-1}\). The first-order condition in \cite[Corollary~3.8]{jindal2024estimates}, which specialises to \eqref{eq:lsr1} in the \LSR{} setting, yields
\begin{equation}
    r_v>pr_u^p. \label{eq:counter1}
\end{equation}
For the second-order bounds of \Cref{thm:imft2}, \(K_{uv}=K_{vv}=0\) and \(K_{uu}=p(p-1)R_u^{p-2}\). Since \eqref{eq:imft2.2} holds automatically, \eqref{eq:imft2.1} yields
\begin{equation}
    r_v>\tfrac12p(p-1)R_u^{p-2}r_u^2. \label{eq:counter2}
\end{equation}
The ratio of the second-order threshold to the first-order threshold is \(\tfrac{p-1}{2}(R_u/r_u)^{p-2}\). Since \(0<r_u<R_u\), this ratio equals \(1/2\) for \(p=2\) and exceeds \(1\) for every \(p\ge3\). Thus, the second-order region is strictly larger for \(p=2\), whereas the first-order region is strictly larger for every \(p\ge3\). \Cref{cor:degenerate} therefore cannot be strengthened to an unconditional statement: the two families of bounds are complementary certificates to be compared case by case, rather than successive refinements of one another.

\subsection{Qualitative Differences between the Bounds}
\label{ssec:qualitative}

Beyond the relative tightness examined above, the two bound families differ qualitatively in their structures.

First, in computational structure: the first-order quantities \(\Lpr,\Lpp\) are suprema taken over the neighbourhoods determined by \((\rpr,\rpp)\) being certified in \eqref{def:lsr1-L}, so verifying \eqref{eq:lsr1} at a candidate radius pair requires optimising over that same pair, and obtaining the maximal admissible \((\rpr,\rpp)\) in general amounts to a running search. The second-order quantities \(\Kuu,\Kuv,\Kvv\), on the other hand, are evaluated once over an a priori fixed region \((\Rpr,\Rpp)\), after which \eqref{eq:lsr2} is a static algebraic system with a closed-form solution, as in \Cref{ssec:lsr2-geometry}. \Cref{cor:approx} extends this advantage further, allowing any computable upper bound on \(\Kuu,\Kuv,\Kvv\) to be used in place of the exact suprema without loss of validity. 

Second, in the geometry of the certified region: the quadratic structure of \eqref{eq:lsr2} guarantees a bounded admissible set with explicit finite radii, i.e., \(\rpr^{\sup}\) and \(\rpp^{\sup}\), whereas the first-order region can be unbounded in a direction, as it is for complete graphs in \Cref{thm:lsr1-regular}, where \(\rpp\) is unconstrained.

These distinctions, together with the fact that neither bound dominates the other outside the degenerate regime of \Cref{cor:degenerate}, indicate that the choice between the two is best made on computational, rather than purely analytical, grounds: the first-order bounds are preferable when \(\Lpr,\Lpp\) are cheap to evaluate or optimise, and the second-order bounds are preferable when Hessian estimates are available, or when an explicitly described, easily visualised closed-form region, as in \Cref{fig:conic}, is preferred over a tighter but implicit one.

\section{Numerical Results and Visualisations}
\label{sec:numerics}

We now turn to a numerical investigation of the second-order bounds for the primary bifurcation in the Hopfield-like dynamics \eqref{eq:regular} on regular graphs with \(d=1\). We first outline the computation of the certified radii for both the \emph{exact} and \emph{approximate} cases in the sense of \Cref{cor:approx}.

For the \emph{exact} bounds, the quantities \(\Kuu,\Kuv,\Kvv\) defined in \eqref{def:lsr2-K} depend on the search radii \(\Rpr\) and \(\Rpp\), inducing \(\rpr^{\sup}(\Rpr,\Rpp)\) and \(\rpp^{\sup}(\Rpr,\Rpp)\) as evaluated in \Cref{ssec:conservative}. As discussed in \Cref{ssec:lsr2-R}, the search radii themselves clip the certified domain whenever they are too small.
Hence, for a fixed graph, we define the reported radii \(\rpr^\star\) and \(\rpp^\star\) as the suprema of the solution sets of the fixed-point equations \(\rpr^{\sup}(\Rpr,\Rpp)=\Rpr\) and \(\rpp^{\sup}(\Rpr,\Rpp)=\Rpp\), respectively.
Similarly, for the \emph{approximate} bounds, the reported radii \(\rhopr^\star\) and \(\rhopp^\star\) are the largest solutions of their respective fixed-point equations, obtained using the analytical upper bounds \(\kuu,\kuv,\kvv\), which depend only on \(\Rpr\).

For each \((n,k)\) pair, the reported radii are averaged over up to \(10\) distinct \(k-\)regular graphs on \(n\) vertices; if fewer than \(10\) non-isomorphic graphs exist, all such graphs are included.

\Cref{fig:regular-nk} illustrates how the exact second-order bounds \(\rpr^\star,\rpp^\star\) vary with graph size \(n\) and degree \(k\).
\Cref{fig:regular-2v2} depicts the conservativeness introduced by \Cref{cor:approx} by visualising the ratio \(r^\star/\rho^\star\) between the exact and approximate bounds.

\begin{figure}[pos=ht!]
    \centering
    \includegraphics[width=.9\linewidth]{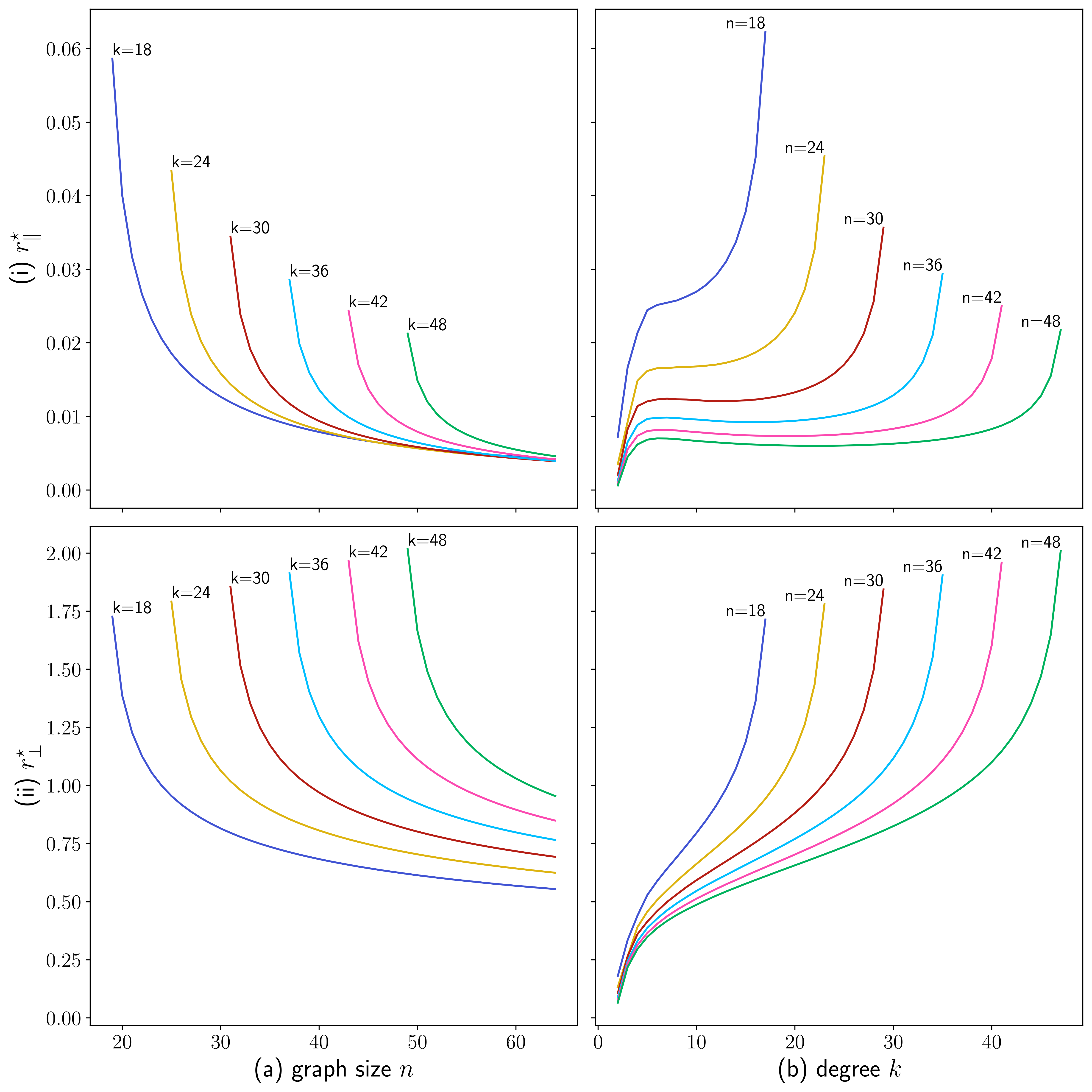}
    \caption{Trends in (i) \(\rpr^\star\) and (ii) \(\rpp^\star\) with graph size and degree: (a) fixed degree \(k\); (b) fixed graph size \(n\).}
    \label{fig:regular-nk}
    \vspace{-5mm}
\end{figure}

\begin{figure}[pos=ht!]
    \centering
    \includegraphics[width=.9\linewidth]{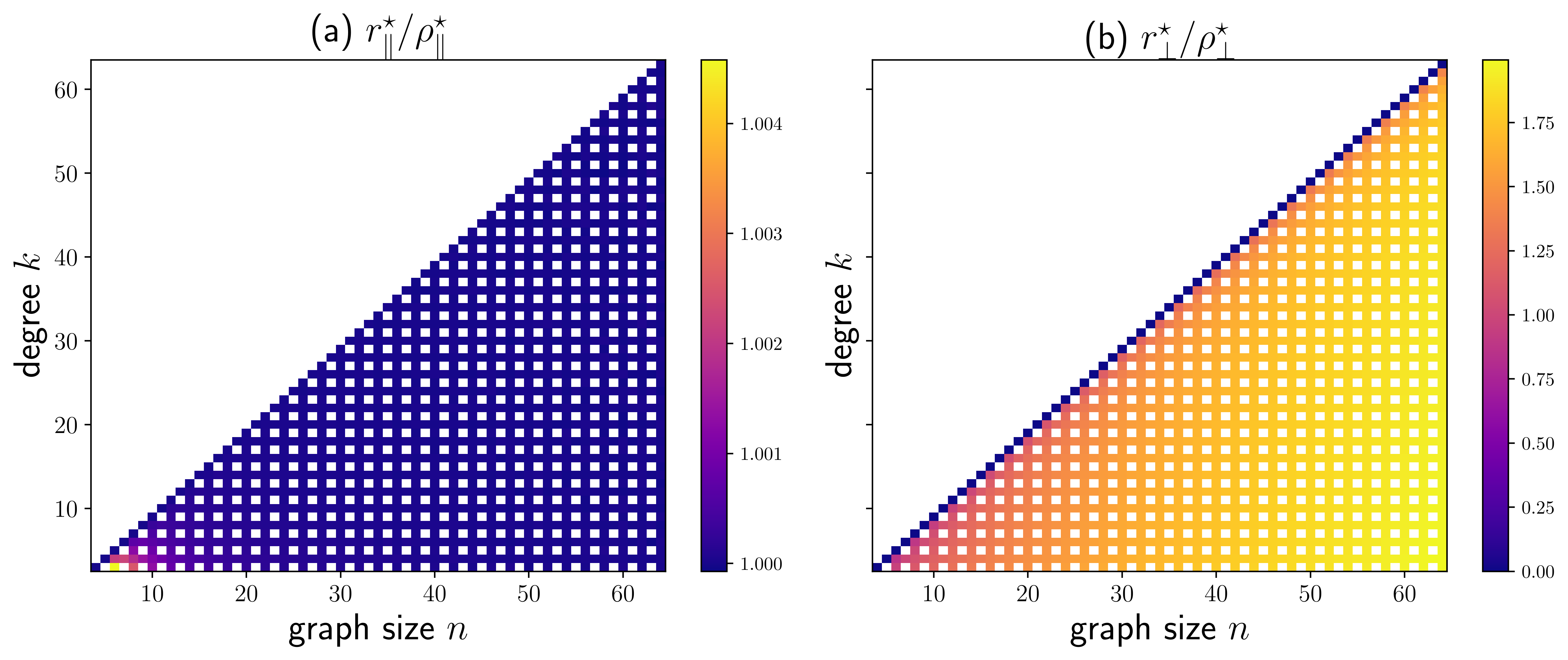}
    \caption{Comparison of the exact and approximate bounds via heatmaps of their ratios (a) \(\rpr^\star/\rhopr^\star\) and (b) \(\rpp^\star/\rhopp^\star\) across graph sizes \(n \le 64\) and admissible degrees \(3 \le k \le n-1\); the approximate bounds are most conservative for sparse graphs (small \(k\)) but track the exact bounds closely otherwise.}
    \label{fig:regular-2v2}
    \vspace{-5mm}
\end{figure}

\appendix

\section{Partial Derivatives of \(\phi\)}
\label{app:dv}

\noindent
This section collects the derivations of expressions for the Jacobians and Hessians of \(\phi(\A,\B,p):=\W^{\top}\Phi^{c}(\A,\B,p)\) w.r.t. \(\A\), \(\B\), and \(p\) in terms of those of \(\Phi(x,p)\) w.r.t. \(x\) and \(p\) and the orthonormal basis matrices \(\V,\Wbar \in \R[n \times q]\) and \(\Vbar,\W \in \R[n \times(n-q)]\) defined in \eqref{eq:SVD}. We use the Einstein summation convention for handling tensors wherever necessary, wherein repeated \emph{dummy} indices are summed over their entire range (both dummy and free indices are subscripted).

\begin{lemma}
    \label{lem:DF}
    The Jacobians of \(\phi\) can be expressed as
    \begin{subequations}
        \begin{alignat}{2}
            \Dp\phi & = \W^{\top}\,\Dp\Phi(x,p)        &  & \in \R[(n-q)\times m],     \label{eq:DpF} \\
            \Da\phi & = \W^{\top}\,\Dx\Phi(x,p)\,\V    &  & \in \R[(n-q)\times q],     \label{eq:DaF} \\
            \Db\phi & = \W^{\top}\,\Dx\Phi(x,p)\,\Vbar &  & \in \R[(n-q)\times (n-q)], \label{eq:DbF}
        \end{alignat} \label{eq:DF}
    \end{subequations} where \(x = \Gamma(\A,\B) = \V\A + \Vbar\B\) as defined in \eqref{def:Gamma}.
\end{lemma}
\begin{proof}
    \eqref{eq:DaF} and \eqref{eq:DbF} follow directly from the chain rule,
    \begin{align*}
        \Da\phi(\A,\B,p) & = \W^{\top}\,\Dx\Phi(x,p)\,\Da x = \W^{\top}\,\Dx\Phi(x,p)\,\V,    \\
        \Db\phi(\A,\B,p) & = \W^{\top}\,\Dx\Phi(x,p)\,\Db x = \W^{\top}\,\Dx\Phi(x,p)\,\Vbar,
    \end{align*}
    while \eqref{eq:DpF} follows directly.
\end{proof}

\begin{lemma}
    \label{lem:D2F}
    The strict partials of \(\phi\) can be expressed as
    \small
    \begin{subequations}
        \begin{alignat}{2}
            \Dp^2\phi & = \W^{\top}\,\Dp^2\Phi(x,p)              &  & \in \R[(n-q)\times m\times m],         \label{eq:DppF} \\
            \Da^2\phi & = \W^{\top}\,\Dx^2\Phi(x,p)[\V,\V]       &  & \in \R[(n-q)\times q\times q],         \label{eq:DaaF} \\
            \Db^2\phi & = \W^{\top}\,\Dx^2\Phi(x,p)[\Vbar,\Vbar] &  & \in \R[(n-q)\times (n-q)\times (n-q)]. \label{eq:DbbF}
        \end{alignat}
    \end{subequations}
\end{lemma}
\begin{proof}
    \eqref{eq:DppF} is obtained by differentiating \eqref{eq:DpF} w.r.t. \(p\) again.
    \eqref{eq:DaaF} follows from differentiating \(\phi_i(\A,\B,p) = \W[ai]\Phi_a(x,p)\) twice with respect to \(\A\) to obtain
    \begin{equation*}
        \bigl[\Da^2\phi(\A,\B,p)\bigr]_{ijk} = \W[ai] \pdv{\Phi_a}{x_b, x_c}(x,p) \pdv{x_b}{\A_j} \pdv{x_c}{\A_k},
    \end{equation*}
    and using \(x_h = \V[hj]\A_j + \Vbar[hj]\B_j \implies \pdv{x_h}{\A_j} = \V[hj]\) to get
    \begin{equation*}
        \bigl[\Da^2\phi(\A,\B,p)\bigr]_{ijk} = \W[ai] \bigl[\Dx^2\Phi(x,p)\bigr]_{abc} \V[bj] \V[ck],
    \end{equation*}
    which is just \(\W^{\top}\,\Dx^2\Phi(x,p)[\V,\V]\) in multilinear notation.
    \eqref{eq:DbbF} follows analogously by using \(\pdv{x_b}{\B_j} = \Vbar[bj]\).
\end{proof}

\begin{lemma}
    \label{lem:DDF}
    The mixed partials of \(\phi\) can be expressed as
    \small
    \begin{subequations}
        \begin{alignat}{2}
            \Da\Db\phi & = \W^{\top}\,\Dx^2\Phi(x,p)[\V,\Vbar]   &  & \in \R[(n-q)\times q\times (n-q)], \label{eq:DabF} \\
            \Db\Da\phi & = \W^{\top}\,\Dx^2\Phi(x,p)[\Vbar,\V]   &  & \in \R[(n-q)\times (n-q)\times q], \label{eq:DbaF} \\
            \Da\Dp\phi & = \W^{\top}\,\Dx\Dp\Phi(x,p)[\V,I_m]    &  & \in \R[(n-q)\times q\times m], \label{eq:DapF}     \\
            \Dp\Da\phi & = \W^{\top}\,\Dp\Dx\Phi(x,p)[I_m,\V]    &  & \in \R[(n-q)\times m\times q], \label{eq:DpaF}     \\
            \Db\Dp\phi & = \W^{\top}\,\Dx\Dp\Phi(x,p)[\Vbar,I_m] &  & \in \R[(n-q)\times (n-q)\times m], \label{eq:DbpF} \\
            \Dp\Db\phi & = \W^{\top}\,\Dp\Dx\Phi(x,p)[I_m,\Vbar] &  & \in \R[(n-q)\times m\times (n-q)]. \label{eq:DpbF}
        \end{alignat}
    \end{subequations}
\end{lemma}
\begin{proof}
    \eqref{eq:DabF} is obtained by differentiating \(\phi_i(\A,\B,p) = \W[ai]\Phi_a(x,p)\) first with respect to \(\B_k\) and then \(\A_j\):
    \begin{equation*}
        \bigl[\Da\Db\phi\bigr]_{ijk} = \W[ai] \pdv{\Phi_a}{x_b, x_c}(x,p) \pdv{x_b}{\A_j} \pdv{x_c}{\B_k}.
    \end{equation*}
    Substituting \(\pdv{x_b}{\A_j}=\V[bj]\) and \(\pdv{x_c}{\B_k}=\Vbar[ck]\) above yields
    \begin{align*}
        \bigl[\Da\Db\phi\bigr]_{ijk} & = \W[ai] \bigl[\Dx^2\Phi(x,p)\bigr]_{abc} \V[bj]\Vbar[ck] \\
        \implies \Da\Db\phi          & = \W^{\top}\,\Dx^2\Phi(x,p)[\V,\Vbar],
    \end{align*}
    in multilinear notation. Reversing the differentiation order,
    \begin{align*}
        \bigl[\Db\Da\phi\bigr]_{ijk} & = \W[ai] \bigl[\Dx^2\Phi(x,p)\bigr]_{abc} \Vbar[bj]\V[ck] \\
        \implies \Db\Da\phi          & = \W^{\top}\,\Dx^2\Phi(x,p)[\Vbar,\V],
    \end{align*}
    yielding \eqref{eq:DbaF}. To show \eqref{eq:DapF}, we differentiate \(\phi_i(\A,\B,p)\) first with respect to \(p_k\) and then \(\A_j\):
    \begin{align*}
        \bigl[\Da\Dp\phi\bigr]_{ijk} & = \W[ai] \pdv{\Phi_a}{x_b, p_k}(x,p) \pdv{x_b}{\A_j}          \\
                                     & = \W[ai] \bigl[\Dx\Dp\Phi(x,p)\bigr]_{abc} \V[bj] \delta_{ck} \\
        \implies \Da\Dp\phi          & = \W^{\top}\,\Dx\Dp\Phi(x,p)[\V,I_m].
    \end{align*}
    \eqref{eq:DbpF} follows analogously, and \eqref{eq:DpaF}, \eqref{eq:DpbF} follow by reversing the order of differentiation.
\end{proof}

\section{Bounding $\sum_{j,k} \left(\sum_{\nu=1}^{n} y_\nu \W[\nu i]\W[\nu j]\W[\nu k]\right)^{2}$}
\label{app:Q}

To bound \(Q_i := \sum_{j,k} \left(\sum_{\nu=1}^{n} y_\nu\W[\nu i]\W[\nu j]\W[\nu k] \right)^2\), we first expand the square and rearrange the sum as follows:
\begin{align*}
    Q_{i} & =\! \sum_{j,k} \sum_{\nu,\rho} y_\nu y_\rho \W[\nu i]\W[\nu j]\W[\nu k] \W[\rho i]\W[\rho j]\W[\rho k]                                            \\
          & =\! \sum_{\nu,\rho} y_\nu y_\rho \W[\nu i]\W[\rho i] \biggl(\!\sum_{j}\W[\nu j]\W[\rho j]\!\biggr)\!\biggl(\!\sum_{k}\W[\nu k]\W[\rho k]\!\biggr) \\
          & =\! \sum_{\nu,\rho} y_\nu y_\rho \W[\nu i]\W[\rho i] (\W\W^{\top})_{\nu\rho}^2.
\end{align*}
Using \(\W\W^{\top} = I - \frac{1}{n}\1[n]\1[n]^{\top} \implies (\W\W^{\top})_{\nu\rho} = \delta_{\nu\rho} - \frac{1}{n}\),
\begin{align*}
    Q_i & = \sum_{\nu,\rho} y_\nu y_\rho \W[\nu i]\W[\rho i] \left(\delta_{\nu\rho}-\frac1n\right)^2                               \\
        & = \left(1-\frac2n\right) \sum_{\nu=1}^{n}y_\nu^2\W[\nu i]^2 + \frac1{n^2} \left(\sum_{\nu=1}^{n}y_\nu\W[\nu i]\right)^2.
\end{align*}
Using \(\sum_\nu\W[\nu i]^2=1\) together with \(\abs{y_\nu} \le m_2\) and applying the Cauchy--Schwarz inequality yields the desired result
\begin{align*}
    Q_i & \le m_2^2\left(1-\frac2n\right) + \frac1{n^2} \left(\sum_{\nu=1}^{n}y_\nu^2\right) \left(\sum_{\nu=1}^{n}\W[\nu i]^2\right) \\
        & \le m_2^2 \left(1-\frac2n\right) + \frac{nm_2^2}{n^2} = m_2^2\left(1-\frac1n\right).
\end{align*}

\bibliographystyle{cas-model2-names}
\bibliography{ref}

\bio{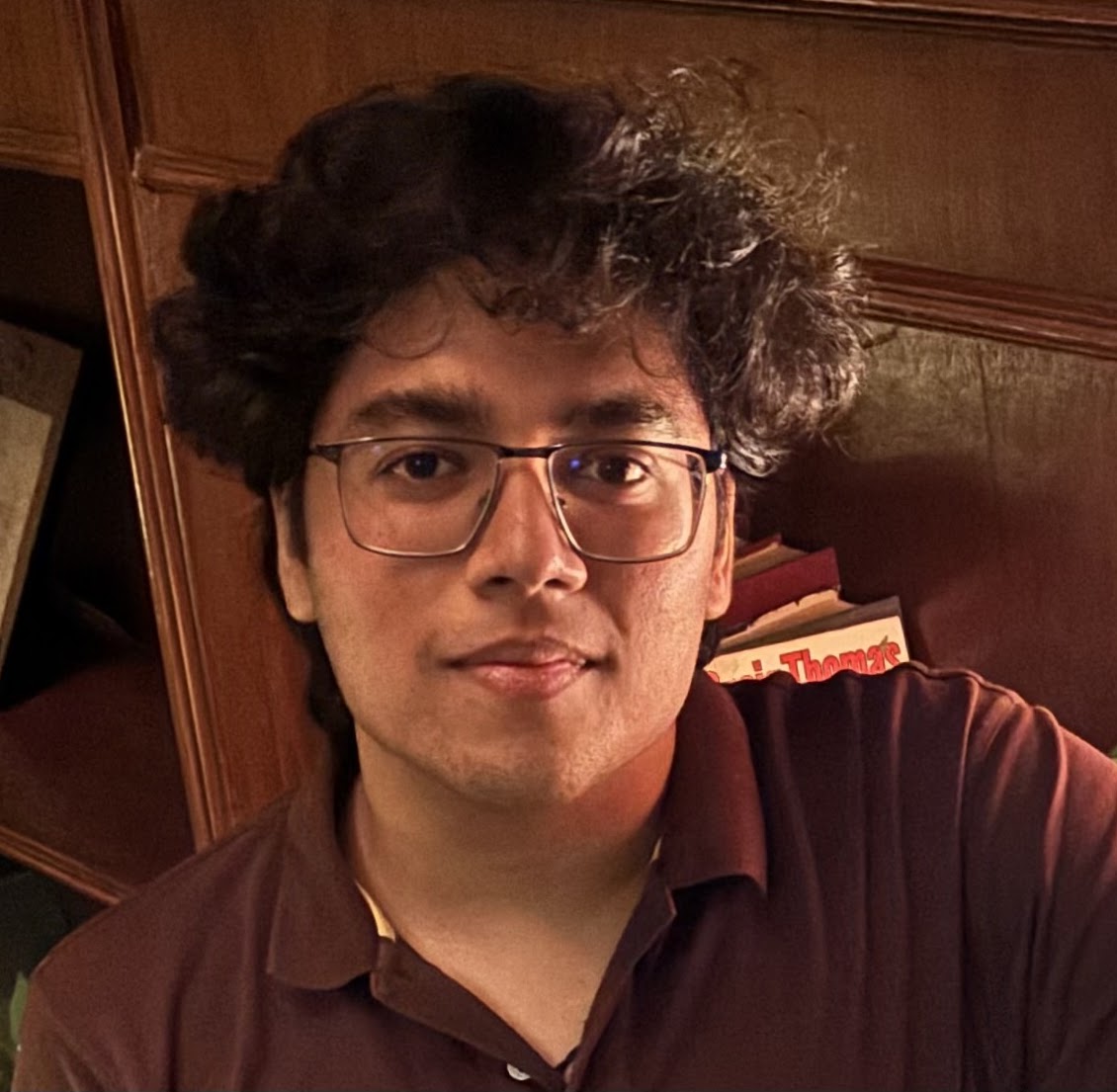}
\textbf{Pranav Gupta} is currently an undergraduate student at IIT Bombay, pursuing a joint Bachelor's degree in Mechanical Engineering and Master's degree in Systems and Control Engineering. He was a visiting research intern at Cornell University, where he worked with Prof. Anastasia Bizyaeva on bifurcations in networked dynamical systems.
\endbio

\bio{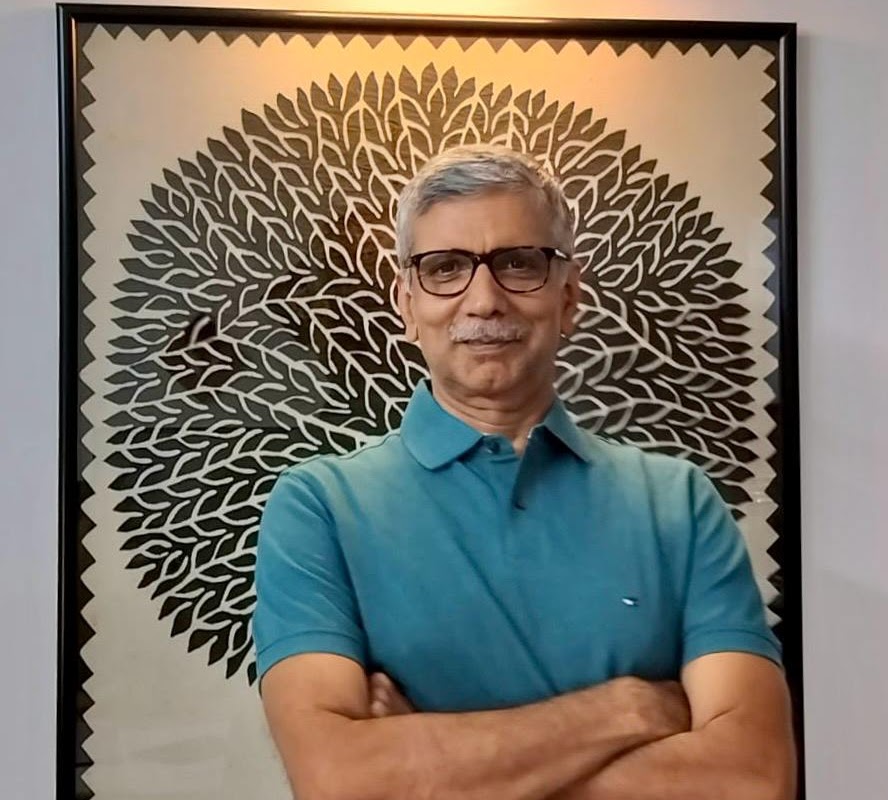}
\textbf{Ravi Banavar} is currently a Professor in Systems and Control Engineering at IIT Bombay. He received his B.Tech. in Mechanical Engineering from IIT Madras (1986), his Masters (Mechanical, 1988) and Ph.D. (Aerospace, 1992) degrees from Clemson University and the University of Texas at Austin, respectively. He had a brief teaching stint at UCLA during 1991--92, soon after which he joined the Systems and Control Engineering group at IIT Bombay in early 1993. He chaired the group from 2009--2015. He has spent a few sabbatical breaks at UCLA (Los Angeles), IISc (Bangalore), IIT Gandhinagar, and LSS (Supelec, France). He currently holds a Guest Professorship at IIT Gandhinagar. He was an Associate Editor of the Elsevier journal, Systems and Control Letters; he is currently an Associate Editor of the Elsevier, IFAC journal, Automatica, and a Technical Associate Editor of the IEEE Control Systems Society magazine.
\endbio

\bio{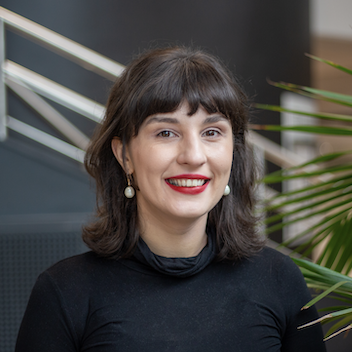}
\textbf{Anastasia Bizyaeva} received the B.A. degree in Physics with a minor in Mechanical Engineering from the University of California, Berkeley, CA, USA, in 2016, and the M.A. and Ph.D. degrees in Mechanical and Aerospace Engineering from Princeton University, Princeton, NJ, USA, in 2019 and 2022, respectively. She is currently an assistant professor in the Sibley School of Mechanical and Aerospace Engineering at Cornell University, where she leads the Control and Computation for Complex Systems Group. Previously, she was a postdoctoral scholar with the NSF AI Institute in Dynamic Systems at the University of Washington. Her research explores mathematical connections between flexible and collectively intelligent behavior in biological and socio-cognitive systems, and the design of autonomous behaviors in engineered teams. She is a recipient of the 2025 IEEE CSS George S. Axelby Outstanding Paper Award.
\endbio

\end{document}